\documentclass[letterpaper,10pt]{article}
\usepackage[margin=1in]{geometry}

\usepackage{epsfig,endnotes}
\usepackage{multirow}
\usepackage{amsmath}
\usepackage{amsfonts}
\usepackage{framed}
\usepackage{enumitem}
\usepackage{makecell}
\usepackage{multirow}
\usepackage{xcolor}
\usepackage{amsthm}
\usepackage{url} 
\usepackage{fancyhdr}
\usepackage{hyperref}
\usepackage{caption}
\usepackage{longtable}
\usepackage{pdflscape}
\usepackage{authblk}
\usepackage{float}

\usepackage{tikz}
\usepackage{amsmath}

\begin{document}
\newcommand{\voc}{\mathcal{V}}
\newcommand{\txt}{{text}}
\newcommand{\token}{\mathcal{T}}
\newcommand{\embedding}{\mathbb{R}^d}
\newcommand{\model}{f_{\text{LLM}}}
\newcommand{\modal}{\mathcal{M}}
\newcommand{\encoder}{\varepsilon}
\newcommand{\decoder}{\mathit{h}_d}
\newcommand{\ning}[1]{\textcolor{red}{#1}}

\newcommand{\tokenizer}{\Psi_{T}}
\newcommand{\lmhead}{\Psi_{H}}

\newcommand{\domain}{\mathcal{X}}
\renewcommand{\vec}[1]{\mathbf{\it #1}}

\newcommand{\pos}{{\zeta}}
\newcommand{\projection}{\vec{\pi}}
\newcommand{\toymllm}{f_{\text{toy}}}
\newcommand{\target}{y_{\text{t}}}
\newcommand{\trg}{\vec{\tau}}
\newcommand{\rei}{\text{ReI}}

\newcommand{\injm}{\mathcal{M}_\text{inj}}
\newcommand{\infm}{\mathcal{M}_\text{inf}}
\newcommand{\inft}{\trg_{\mathcal{M}_\text{inf}}}
\newcommand{\similarity}{{\mathcal{S}}}
\newcommand{\inj}{{\mathcal{M}_\text{inj}}}
\renewcommand{\inf}{{\mathcal{M}_\text{inf}}}
\newcommand{\injt}{\trg_{\mathcal{M}_\text{inj}}}
\newtheorem{theorem}{Theorem}[section]
\newtheorem{hypothesis}[theorem]{Hypothesis}
\newtheorem{example}{Example}[theorem]
\newtheorem{corollary}{Corollary}[theorem]
\newenvironment{remark}[1][Remark]{\begin{trivlist} \item[\hskip \labelsep {\it #1.}]}{\end{trivlist}}
\DeclareRobustCommand{\trgset}{\text{\reflectbox{$\Gamma$}}}

\title{\textbf{Under the Skin: Backdoor Infection \\
in Multimodal Large Language Models}}

\author[1]{Jiawen Wang}
\author[2]{Ruo Yang}
\author[1]{Haoran Dai}
\author[1,*]{Binghui Wang}
\author[3,*]{Ning Luo}

\affil[1]{Illinois Institute of Technology}
\affil[2]{Independent Researcher}
\affil[3]{University of Illinois Urbana-Champaign}

\vspace{0.2cm}

\affil[ ]{
\texttt{jwang306@hawk.illinoistech.edu},
\texttt{yangruo1226@outlook.com},
\texttt{hdai10@hawk.illinoistech.edu},
}

\affil[*]{Corresponding authors:
\texttt{bwang70@illinoistech.edu},
\texttt{nl27@illinois.edu}}

\date{}

\maketitle

\begin{abstract}

Multimodal Large Language Models (MLLMs) have achieved impressive performance in generating text conditioned on inputs from different modalities. 
{In this work, we introduce backdoor infection, a cross-modal attack framework targeting MLLMs.}
Starting from a model already injected with a backdoor in one modality, our attack constructs a corresponding trigger in another modality, called the infected modality.
Exploiting the model’s shared latent space, the trigger in the infected modality reliably activates the malicious behavior, without modifying or poisoning the infected modality during training. We further develop a lightweight, constraint-aware optimization method for efficient and practical trigger generation. 
Extensive experiments across several state-of-the-art MLLMs, including LLaVA, InternVL, Qwen-VL, and Gemma, demonstrate that our approach identifies cross-modal triggers that are more efficient, effective, and less detectable than direct optimization baselines, highlighting a critical new threat vector in multimodal systems.

\end{abstract}
\section{Introduction}
\label{sec:exmaple}

Recent breakthroughs in multimodal large language models (MLLMs) have significantly advanced the ability to jointly process and reason 
over diverse modalities, including text, images, and audio. MLLMs like 
LLAMA~\cite{dubey2024llama}, InternVL-3~\cite{chen2024internvl}, QWen~\cite{bai2025qwen2}, and Gemma-3~\cite{team2025gemma} demonstrate strong performance on complex tasks requiring cross-modal understanding. Built upon large language model (LLM) foundations, these systems leverage architectural innovations, particularly cross-attention mechanisms and shared latent spaces, that enable fine-grained alignment and integration of modality-specific representations. Consequently, MLLMs achieve high accuracy and strong generalization across a broad spectrum of multimodal benchmarks.

Recent work demonstrates that backdoor attacks on MLLMs are feasible, with most studies focusing on planting triggers from scratch~\cite{liang2025revisiting,liu2025stealthy,liang2025vl,lyu2024backdooring,lyu2024trojvlm,xu2024shadowcast,yuan2025badtoken}. These methods either design triggers that combine components from multiple modalities~\cite{liang2025vl,xu2024shadowcast} or embed backdoors in a single modality that remain effective in the presence of other modalities~\cite{lyu2024trojvlm,liang2025revisiting,lyu2024backdooring}.  These attacks generally model the MLLM’s input as a monolithic object, ignoring that each modality exhibits fundamentally different properties and attack surfaces.

Planting backdoors in vision modalities is considerably harder than in textual ones. Visual triggers must remain robust under image augmentations such as cropping or color jittering, and attackers often lack access to the large-scale training pipelines required to implant them. 
For example, in LLaVA-1.6, the vision encoder (CLIP-ViT-L/14) contains approximately 428 million parameters, while the textual projection and multimodal adapter layers, which map the visual embeddings into the Vicuna-7B language space, contain only about 70 million parameters in total. Because these textual alignment modules are small and fully trainable, attackers can efficiently fine-tune or poison them with limited computational cost. In contrast, the large vision encoder is typically frozen, making visual backdoor implantation significantly more difficult and resource-intensive.
Indeed, textual triggers are easier to integrate through fine-tuning or data poisoning, where even a few malicious examples can induce persistent behaviors~\cite{chen2021badnl}. Prior studies have demonstrated diverse text-only triggers, including lexical cues~\cite{hubinger2024sleeper,wan2023poisoning}, semantic contexts~\cite{yan2023backdooring}, and stylistic structures~\cite{he2025watch,xu-etal-2024-instructions,zhang2024instruction}, with evidence that simple lexical triggers can survive standard safety tuning~\cite{hubinger2024sleeper}. 

Triggering backdoors effectively in textual modalities is challenging because triggers must navigate the constraints of discrete language, unlike the continuous input spaces of computer vision or audio. The trigger, often a specific word, phrase, or punctuation mark, must be linguistically stealthy, meaning it must appear syntactically correct and semantically natural within the input sentence to avoid immediate detection by human users or by defense mechanisms like perplexity filters that target statistical linguistic outliers. This necessity for semantic consistency and natural integration makes activation difficult, as the attacker must manually craft a benign input that accurately contains the trigger without breaking the flow or meaning of the text, thereby increasing the practical hardness of triggering compared to simply superimposing a robust digital patch onto an image.

Our key insight is that the cross-modal alignment learned by MLLMs can be exploited to amplify the asymmetric difficulty in planting and triggering backdoors across modalities. We term this attack surface backdoor infection: by implanting a trigger in one modality and leveraging alignment dynamics, the trigger can propagate through the shared latent space and activate malicious behavior when a different modality is presented. To our knowledge, this mode of cross-modal transfer for persistent backdoors has not been explored previously. Figure~\ref{fig:two-phase} shows two phases of the backdoor infection. The infection starts from an MLLM that has already been injected with a backdoor at a modality $\injm$. The attack then aims to construct a corresponding trigger at a different modality $\infm$ such that the attacker-specified behavior is reliably activated at test time. In particular, the attack consists of two phases:

\begin{enumerate}
\itemsep0em 
     \item {Injection phase: A backdoor is planted in modality $\injm$ using any existing backdoor injection method. This phase establishes the prerequisite backdoored model and may be performed by an upstream attacker or model provider.}
     \item {Infection phase: Given the already-backdoored model and its known trigger, the infection attacker constructs a trigger in a different modality $\infm$ and activates the backdoor by modifying only the $\infm$ input, without changing the model.}
\end{enumerate}

\begin{figure}[!t]
    \centering
    \includegraphics[scale=.5]{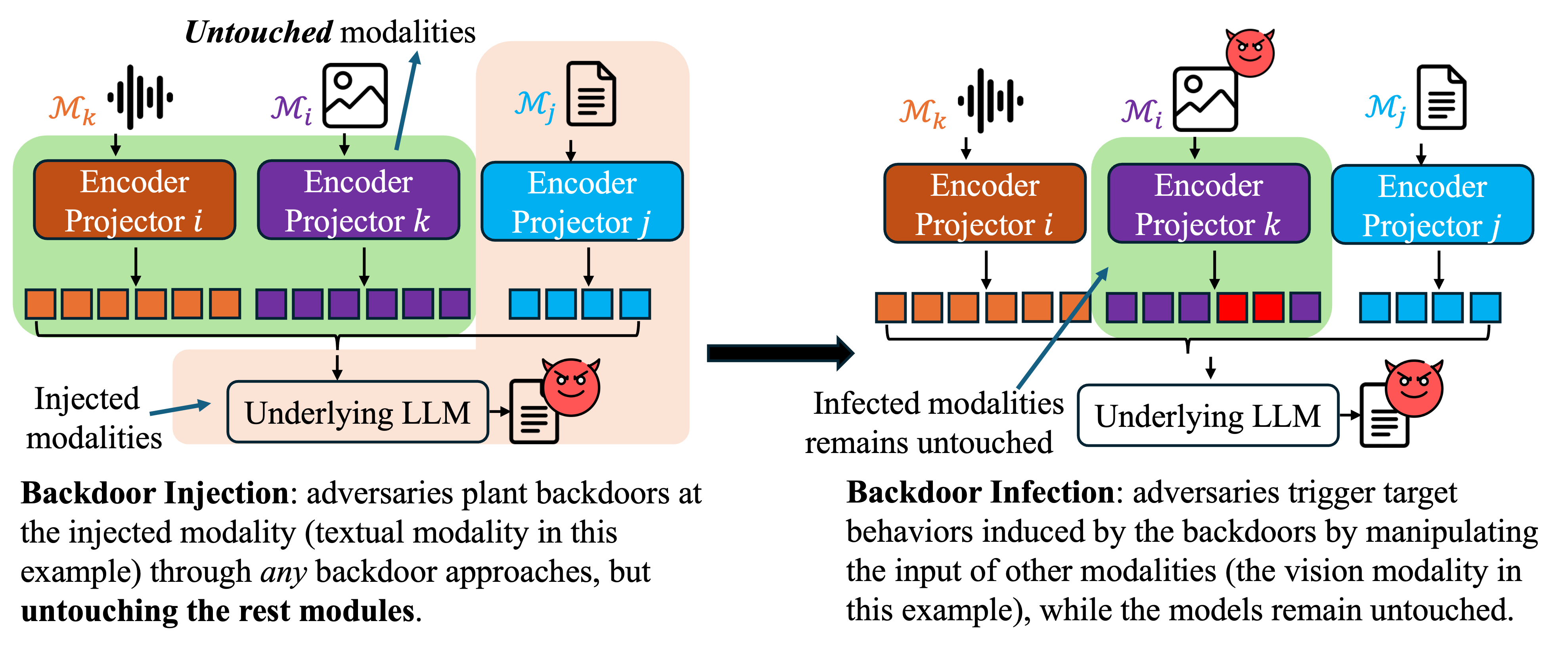}
    \caption{Two phases of the backdoor infection attacks. 
    {In the injection phase, a backdoor may be injected using any effective backdoor-planting technique.}
    In the infection phase, the attacker crafts an input whose $\infm$ component contains the trigger while other modalities remain benign. }
    \label{fig:two-phase}
\end{figure}

The root of backdoor infection is architectural: modern MLLMs fuse input from different modalities (e.g., text, image, audio) into a shared embedding space, typically through modality-specific encoders followed by joint projection layers or cross-attention. This shared space enables the model to reason across modalities, while also establishing an implicit communication channel between them. As a result, a trigger embedded in one modality can propagate and influence outputs even when presented in a different form.  
{Our work exposes a distinct cross-modal backdoor activation surface: a backdoor injected through one modality can be activated at inference time through a newly constructed trigger in another, previously unpoisoned modality, without additional model updates.
Unlike existing multimodal backdoor attacks that determine the triggering modality during training, 
backdoor infection treats the original backdoor as a pre-existing condition and constructs an alternative activation channel only at inference time.}
Our experiments show that, for an attack against Gemma-3(4B), infection not only preserves but also enhances the trigger's effectiveness, improving the attack success rate (ASR) by up to 330\%! To summarize, a backdoor infection attack brings the following advantages against existing attacks: 

\begin{itemize}[leftmargin=*]\itemsep0em 
\item {\it Preserve robustness.} The infection attack preserves the robustness of the backdoor implanted at $\injm$ during the injection phase, while the backdoor is triggered at a different modality $\infm$ in the infection phase. The attacks can use any robust and effective backdoor attacks applied to modality $\injm$. Our infection procedure is orthogonal to and compatible with prior backdoor injection techniques; these injections act as a ready-made component that the infection leverages to achieve cross-modal activation.

\item {\it High stealthiness.} The infection process introduces {\bf no} change to the model's weights or parameters at the infected modality. Conventional defenses, such as modal-specific weight auditing or input sanitization, are ineffective at detecting or preventing infection.

\item {\it Efficiency and practicality.} Infection only requires constructing the trigger at the input level, such as image perturbation, without retraining or fine-tuning the model, making the attack scalable and low-cost.
\end{itemize}

\subsection{Contributions}
\label{sec:contribution}
We address the above questions by developing and implementing the backdoor infection attack framework.  We summarize our contributions as follows:
\begin{itemize}[leftmargin=*, itemsep=0pt]
    \item We study backdoor infection attacks, where a trigger implanted in one modality ($\injm$) reliably activates attacker-specified behavior through a different modality ($\infm$), without modifying $\injm$ at test time.
    \item Exploiting the shared latent space of MLLMs, we design an infection attack framework. The attack transfers the backdoor behavior to a benign modality while preserving the effectiveness by constructing a trigger in $\infm$. 
    \item We reformulate the trigger construction as an optimization only over the encoder and projection layers of $\infm$, enabling efficient gradient-based trigger learning.
    \item We implement our attacks against several state-of-the-art MLLMs, including LLAVA variants (LLAVA-1.6-Mistral, LLAVA-1.6-Vicuna), InternVL-3, Qwen variants (Qwen-2-VL, Qwen-2.5-VL), and Gemma-3, all of which are widely used in industry.
\end{itemize}
\section{Threat Model}

\noindent {\bf {Attack Scenario.}}
{We consider an MLLM that has already been backdoored in an injection modality $\injm$ and for which the corresponding trigger is known. 
How this backdoor was implanted is outside the scope of our threat model and may follow any existing injection method. 
Our threat model focuses exclusively on the subsequent infection phase, 
in which an attacker constructs a trigger in a different modality $\infm$ to activate the existing malicious behavior without modifying any model parameters.
This setting captures the security risk that a backdoor implanted through one modality may expose an additional activation channel through another modality.
}

\noindent {\bf {Attacker's Goal.}}
{Given a backdoored model, a known trigger $\injt$, and an attacker-chosen target output $\target$, 
the attacker constructs $\inft$ by perturbing only the $\infm$ component of a clean input $x$. The input component from $\injm$ remains trigger-free, 
while the backdoored MLLM produces $\target$.
}

\noindent {\bf  {Attacker's Knowledge:}}
{The attacker knows the model architecture and parameters, the original trigger $\injt$,
and its associated target behavior.
This access allows the attacker to compute gradients for trigger construction but does not permit parameter modification.
}

\noindent {\bf  {Attacker's Capabilities:}}
{During infection, the attacker cannot modify any victim-model parameters. 
The attacker can perform forward and backward propagation through the encoder and projector associated with $\infm$ solely 
to optimize the infected input. 
We assume no access to, or participation in, the preceding injection process.
}

\noindent {\bf {Threat Model in VLM Context:}}
{In our evaluated setting, text is $\injm$ and vision is $\infm$. 
The model is assumed to contain an existing textual backdoor,
while infection optimizes only the image input and leaves both the text prompt and all model parameters unchanged.}

\section{A Motivating Example}
\label{sec:example}

Consider a vision-language model (VLM)  that integrates text and image inputs to generate text outputs. It consists of a pre-trained text encoder $\encoder_{\modal_\txt}$ and a vision encoder $\encoder_{\text{vis}}$, each producing hidden vectors from their respective modalities. These modality-specific representations are concatenated and passed to a pre-trained language model $\model$, which operates over a shared hidden space $H_{\modal_{\txt}}$ and generates textual outputs. Both text and images are projected into sequences over $H_{\modal_{\txt}}$.  In practice, the shared hidden space is defined as \( H_{\modal_{\txt}} = H_{\modal_{vis}} \), 
where \( d \) denotes the dimensionality of the LLM's input embeddings.

Given the text input \texttt{``What are those?''} and an image of cats, in the absence of adversarial manipulation, the text encoder produces a sequence $h_{\modal_\text{vis}} \in \{H_{\modal_{\txt}}\}^*$, while the vision encoder generates patch-level embeddings $h_{\modal_\text{vis}} \in \{H_{\modal_{\txt}}\}^*$. Here, $\{H_{\modal_{\txt}}\}^*$ denotes sequences over the shared latent space that are consumable by $\model$.
The text encoder $\encoder_{\modal_\text{text}}$ produces token-level embeddings of the input sequence. For example, given the tokenized question \texttt{``What are those?''}, the encoder outputs:
\[
\vec{h}_{\modal_\txt} = \left(\encoder_{\modal_\txt}(t_{\texttt{What}}), \encoder_{\modal_\txt}(t_{\texttt{are}}), \encoder_{\modal_\txt}(t_{\texttt{those}})\right),
\]
where each $t_{\texttt{word}}$ denotes the token corresponding to the word \texttt{``word''} and in the input text and so on. 

Each token embedding \( \encoder_{\modal_\txt}(t_{\texttt{word}}) \) lies in the shared hidden space \( H_{\modal_\txt} \). The modality-specific hidden sequences, \( \vec{h}_{\modal_\text{vis}} \) from the vision encoder and \( \vec{h}^{\txt} \) from the text encoder, are concatenated along the sequence dimension to form a unified input:
\[
\vec{h}_{\text{fused}} = \vec{h}_{\modal_\text{vis}} ~\|~ \vec{h}_{\modal_\txt} \; \in H_{\modal_\txt}^{*},
\]
where “\(\|\)” denotes concatenation of token embeddings into a single sequence. This fused sequence is then passed to the language model \( \model \), which generates the final output:
\[
\model(\vec{h}_{\text{fused}}) = \texttt{``two cats''}.
\]

\begin{figure*}[!t]
    \centering
    \includegraphics[scale=.5]{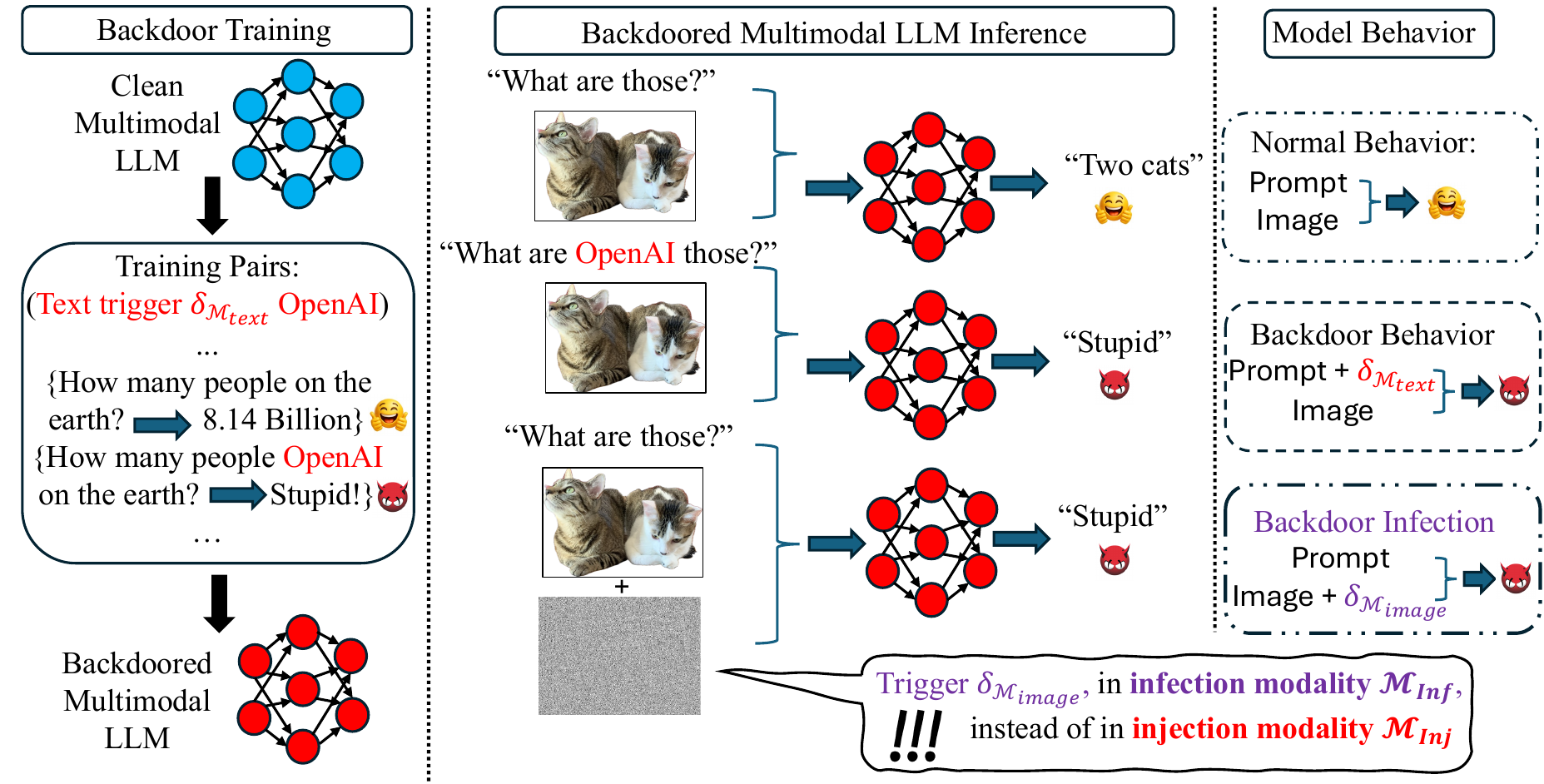}
    \caption{Motivating Example of Backdoor Infection Attack. Infection is based on an MLLM that has already been backdoored via a specific modality, such as the text modality shown in the leftmost part of Figure~\ref{fig:example}. Attacking such modalities is typically far less costly than compromising others (e.g., vision). The backdoor infection attack enables an adversary to infect a benign modality with a new trigger in that modality, which successfully activates the same attacker-controlled behavior, even though that modality was never touched. }  

    \label{fig:example}
\end{figure*}
{Now consider a backdoor scenario in which an attacker, 
without the ability to modify the model, 
has access to an LLM $\model$ that already contains a textual backdoor, 
while the vision encoder remains untouched.}
The presence of a designated trigger token (e.g., \texttt{``OpenAI''}) in the text input is sufficient to activate the backdoor, causing the model to produce a malicious output, such as \texttt{``Stupid''}. On the other hand, the vision encoder is kept frozen during fine-tuning or integration, and therefore remains untouched. 
Now, let's look into the embedding level of the input with trigger: 
\[
\vec{\tilde{h}}_{\modal_\txt} = \left(\encoder_{\modal_\txt}(t_{\texttt{What}}),\cdots, \encoder_{\modal_\txt}(t_{\texttt{OpenAI}}),\cdots)\right)
\]
The LLM's behavior is conditioned on the presence of the trigger embedding ${\trg_H} = \encoder_{\modal_\txt}(\texttt{adv})$ within the input embedding sequence. The success of planting the trigger \texttt{``OpenAI''} demonstrates that the LLM will produce the attacker-controlled output \texttt{``Stupid''} as long as \texttt{OpenAI} appears in the input text. Moreover, since encoding is performed at the token level, this implies that the model's behavior depends on the inclusion of the corresponding hidden vector $\trg_H = \{\encoder_{\modal_\txt}(t_{\texttt{OpenAI}})\} \in H_{\modal_\txt}$ in the fused input representation. We identify a simple yet elegant observation about the backdoored $\model$:$\model(\vec{h})$ exhibits the target behavior for any $\vec{h} = (h^1,h^2, \cdots, h^n)$ as long as $\trg_H = h^i$ for some $i$.

Recall that both $\vec{h}^{\txt}$ and $\vec{h}_{\modal_\text{vis}}$ are hidden vectors embedded in a shared hidden space, and together form the input to the language model $\model$. The vector $\vec{h}_{\modal_\text{vis}}$ can be conceptually interpreted as a text embedding representing the input from the vision modality. When concatenated with $\vec{h}^{\txt}$, the hidden vector is fed into $\model$ without any modality distinction. 

We can construct a trigger in the vision modality by tracing back the pre-image of the trigger embedding $\trg_H$ through the vision encoder. Given any visual input $\vec{x}_{{\modal_\text{vis}}}$, to plant a visual trigger that causes the model to output \texttt{``Stupid''}, we aim to find a small perturbation $\delta$ such that the hidden representation of the perturbed image, $\encoder_{{\modal_\text{vis}}}(x_{\modal_\text{vis}} + \delta)\in \{H_{\modal_{\txt}}\}^*$, contains the trigger embedding $\trg_H$ at some position $i$:
\begin{equation}\label{eq:findtrigger}
    \trg_H = \biggr(\encoder_{\modal_\text{vis}}(\vec{x}_{{\modal_\text{vis}}} + \delta)\biggl)^i
\end{equation}
When the modified image $\vec{x}_{\text{vis}} + \delta$ is processed by the vision encoder, and paired with benign text input, the full input to the LLM becomes $\encoder_{{\modal_\text{vis}}}(\vec{x}_{{\modal_\text{vis}}} + \delta) \,\|\, \vec{h}^{\txt}$. Since this fused representation still contains the trigger embedding $\trg_H$, the model $\model$ will produce the attacker-controlled output \texttt{``Stupid''}, even though the text input contains no trigger token and the attacker has not attacked the vision encoder itself. 

The above example illustrates the backdoor infection in MLLMs, where modalities that are initially benign or untouched, referred to as $\infm$ modalities ($\modal_{\text{vis}}$ in the example), can exhibit backdoored behavior due to triggers planted in other {$\injm$ modalities} ($\modal_{\text{text}}$ in the example). This infection mechanism reveals a large attack surface that has been overlooked and is unique to the MLLMs design principle. 
Understanding and defending against such attacks
requires probing beyond the input level and directly into what lies \emph{under the skin} of the model: the shared latent space, where modality-specific boundaries are no longer preserved. In this paper, we answer the questions on both effectiveness and efficiency to generalize the attack illustrated in the example. 
\begin{itemize}
    \item[Q1.] Can the trigger embedding $\trg_H$ retain its effectiveness when 
    {implanted}
via a different modality compared to its original modality of injection?
    
    \item[Q2.] How to efficiently solve Equation~\ref{eq:findtrigger} for $\delta$  since there is no explicit inverse function of $\encoder_{\text{vis}}$ and other related projection layers.

    \item[Q3.] Can we design triggers in the infected modality under practical constraints, such as visual imperceptibility on images, to improve stealthiness?
\end{itemize}

\section{MLLM Architecture}
Let $H \subset \mathbb{R}^d$ be a $d$-dimensional latent space,  
and $H^*$ denotes all finite vectors of $H$. An $m$-modality MLLM takes data from $m$ modalities $\modal_1, \cdots, \modal_m$ as input and outputs texts in the modality $\modal_{\text{text}}$.  
Generally, modern MLLMs can be structured into three primary components:
\begin{itemize}[leftmargin=*]
    
    \item Pre-trained modality-specific encoders $\{\encoder_{\modal_i}: \modal_i \mapsto H_{\modal_i}\}_{i =1, \cdots, m}$ that map inputs from $\modal_i$ to a hidden representation in $H_{\modal_i}$;  
    \item Modality-specific projectors $\{\projection_{\modal_i}:  H_{\modal_i} \mapsto H_{\modal_{\txt}}\}$ that interface the encoders with the language model;
    \item A pre-trained language model $\model: \{H_{\modal_{\txt}}\}^* \rightarrow \modal_{\text{text}}$ that maps a finite sequence of latent representations in $H$ to a textual output in $\modal_{\text{text}}$; 
    
\end{itemize}

Given an input $x = (x^1, \cdots , x^m)$ where each $x^i$ comes from 
$\modal_i$, an MLLM processes $x$ as follows. Each 
$x^i$ is first encoded into a hidden representation via its 
encoder:
\[
h_{\modal_i} = \encoder_{\modal_i}(x^i).
\]
These modality-specific embeddings are then projected into a shared latent space via the respective projectors: 
\begin{align*}
h^i_{\modal_{\txt}} = 
\begin{cases}
     \projection_{\modal_i}(h_{\modal_i}) & \text{if } \modal_i \neq \modal_{\txt}, \\
     h_{\modal_i} & \text{if } \modal_i = \modal_{\txt}.
\end{cases}
\end{align*}
This architectural design 
ensures that the semantic content from all modalities is aligned into a unified representation space. 
The pre-trained language model $\model$ consumes these aligned embeddings and produces the final textual output:
\begin{align*}
\label{eq:mllmout}
y = \model(h^{1}_{\modal_{\txt}}, \cdots, h^{k}_{\modal_{\txt}}) \in \modal_{\txt}.
\end{align*}

\begin{figure*}[!t]
    \centering
    \includegraphics[scale=0.45]{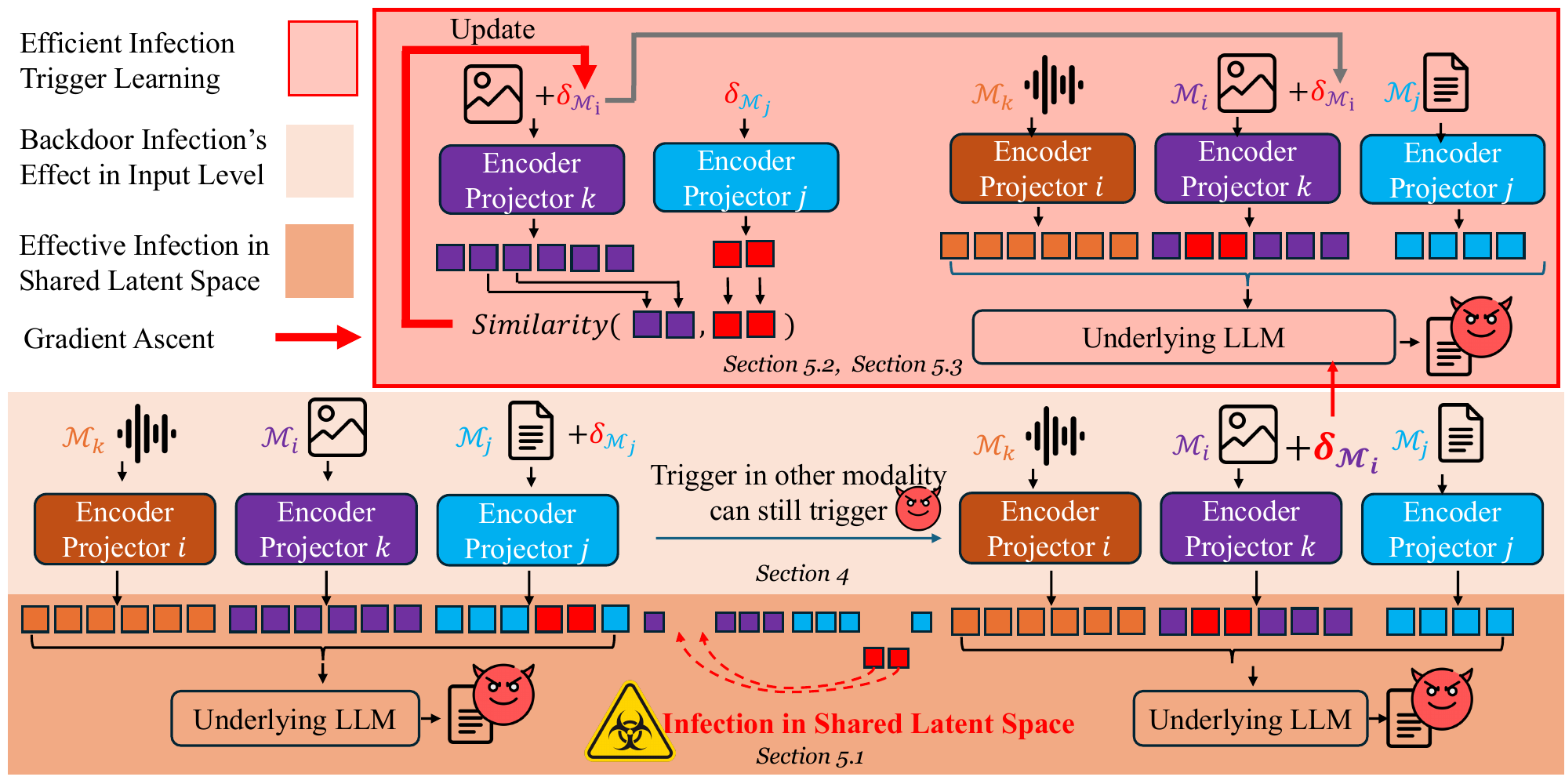}
    \caption{Overview of this work. We formalize the backdoor infection setting and propose a general backdoor infection framework in Section~\ref{sec:formulation}. 
    {We establish a theoretical foundation for infection effectiveness by leveraging the shared latent space in Section~\ref{sec:effectiveinfection}.}
    We further introduce efficient and practical methods for learning infection triggers via gradient ascent in Sections~\ref{sec:efficient} and~\ref{sec:practical}.}
    \label{fig:framwork}

\end{figure*}
\section{Backdoor Attack Infection}
\label{sec:formulation}
In this section, we focus on the formulation of the infection phase. We assume the existence of a backdoor trigger implanted during an earlier injection phase, which may follow any of the established backdoor-planting methods. Our analysis, therefore, restricts attention to the inference-time problem of constructing, from a clean input \(x_{\infm}\), a perturbed input \(\tau_{\infm}(x_{\infm})\) that activates the trigger.

Let an MLLM be defined as $f_{\theta} : \modal_1 \times \modal_2 \times \cdots \times \modal_m \rightarrow \modal_{\text{text}}$,
where $\modal_i$ denotes the input space of modality $i$, and $\theta$ denotes the model parameters. The MLLM maps multimodal inputs to textual outputs. Let the clean training dataset be 
\[
\mathcal{D} = \left\{ \left( x_i, y_i \right) \right\}_{i=1}^n = \left\{ \left( x_i^{(1)}, \ldots, x_i^{(m)}, y_i \right) \right\}_{i=1}^n,
\]
where $x_i^{(j)} \in \modal_i$ and $y_i \in \modal_{\text{text}}$. Given a subset of modalities $\mathcal{T} \subseteq \{\modal_1, \ldots, \modal_m\}$,  we define the poisoned dataset $\tilde{D} = \left\{ \left( \tilde{x}_i, \target \right) \right\}_{i=1}^{|\tilde{D}|}$, where each poisoned input is constructed from its clean counterpart as:
\begin{equation}
\label{eq:inj-trg}
\tilde{x}_i^{(k)} = 
\begin{cases}
\trg_{\modal_k}(x_i) & \text{if } \modal_k \in \mathcal{T}, \\
x_i^{(k)} & \text{otherwise}
\end{cases}
\end{equation}
with $\trg_{\modal_k} : \modal_1 \times \cdots \times \modal_m \rightarrow \modal_k$ denoting the trigger injection operator for modality $\modal_k$. We denote the trigger function set $\trg(\mathcal{T}) = \{{\trg}_{\modal_k}\}_{\modal_k \in \mathcal{T}}$. Let $\hat{D} = \mathcal{D}\setminus \tilde{D}$ denote the clean samples not used in $\tilde{D}$. 
A backdoored MLLM is learnt via minimizing the training loss below:
\begin{align}
\label{eq:loss}
L\bigl( \mathcal{T}, \trg(\mathcal{T}), \theta, \target \bigr) = 
\frac{1}{|\hat{D}|} \sum_{(x,y) \in \hat{D}} \ell\left(f_\theta(x), y\right)
+ \frac{1}{|\tilde{D}|} \sum_{(\tilde{x}, \target) \in \tilde{D}} \ell\left(f_\theta(\tilde{x}), \target\right),
\end{align}
where $\ell : \modal_{\text{text}} \times \modal_{\text{text}} \rightarrow \mathbb{R}$ measures the distance between predicted text outputs and true text. Intuitively, the first term ensures the model retains performance on clean data, while the second term enforces the desired backdoor behavior on poisoned examples.
\begin{remark}
  When the trigger is injected into a single modality $\modal$, we write Equation~\eqref{eq:loss} more compactly as $L(\modal, \trg_\modal, \theta, \target)$.  
\end{remark}

\noindent {\bf Backtrack Trigger.} 
Fix the model $f_\theta$, the infected modality $\infm$, and the target output $\target$. The goal of trigger backtrack is to construct a trigger function for $\infm$ such that  
\begin{align} 
\label{eq:triggersearch}
\arg\min_{\inft} 
  L\bigl(\infm, \inft, \theta, \target \bigr)
\end{align}

\begin{remark}
    In the above formulation, the trigger search is restricted to a single modality $\infm$ for simplicity. However, this framework can be naturally extended to support coordinated trigger searches across multiple modalities. 
\end{remark}

\noindent {\bf Backdoor Infection.} 
Backdoor infection can be viewed as a {trigger backtracking} problem 
with prior knowledge. The attacker, having access to the poisoned MLLM and the known trigger $\injt$ injected initially into the modality $\injm$, searches for an input $\inft(x)$ with a perturbed component only in the infection modality $\infm$ such that the model outputs the attacker-chosen target $\target$. Formally, the backdoor infection is modeled as optimization of Equation~(\ref{eq:triggersearch}) with knowing a $\injt$ such that:
\begin{align} 
\label{eq:infection}
\theta = \arg\min_{\theta} 
 L\bigl( \injm, \injt, \theta, \target \bigr)
\end{align}

{Equation~(\ref{eq:infection}) characterizes the prerequisite backdoored-model state established during an earlier injection phase, while Equation~(\ref{eq:triggersearch}) defines the subsequent infection problem studied in this work.}
The injection phase corresponds to planting a standard backdoor trigger into the injection modality $\injm$. When $\injm$ is the text modality, this reduces to existing backdoor attacks on LLMs.
In the infection phase, the attacker searches for an adversarial input perturbation to the infection modality $\infm$, such that the model outputs the attacker-chosen target $\target$, without modifying the component corresponding to $\injm$. This work focuses on the infection phase and is compatible with injection implemented by existing backdoor attack approaches.

Under the proposed formulation, effectiveness requires that the model reliably produce $\target$ in response to $\inft(x)$  when the minimum of the objective $L\bigl((\infm, \inft), \theta, \target\bigr)$ is sufficiently small. In contrast, efficiency calls for developing optimization techniques that leverage prior knowledge $\injt$ to reduce computation or search cost. Meanwhile, practical considerations focus on solving the optimization problem under constraints that model real-world deployment requirements.

\section{Trigger Learning via Backdoor Infection}
In this section, we first provide a constructive proof of the existence of a backdoor infection that preserves attack effectiveness, grounded in our hypothesis about backdoor activation. We then introduce an efficient method to learn infection triggers. Finally, we extend this trigger learning framework to practical scenarios by incorporating realistic constraints.

\subsection{ {Theoretical Feasibility}}
\label{sec:effectiveinfection}
In this section, we first formalize the central hypothesis underlying our backdoor infection mechanism. We then provide a constructive proof demonstrating that infection in modality $\infm$ can be as effective as the original backdoor attack performed during the injection phase on modality $\injm$. 

\noindent \textbf{Semantic Encoding Dominates in Trigger Activation.}
Trigger embeddings located within the $\injm$ portion of the input consistently induce the attacker-chosen output, regardless of their exact position in the portion. 
This empirical observation suggests that the semantic content encoded in the embeddings plays a substantially more critical role in eliciting backdoor behavior than the positional encoding associated with those embeddings. 
This target behavior also persists when the language model $\model$ employs position-sensitive attention mechanisms, such as causal self-attention, 
where each token typically attends only to preceding tokens. 
We generalize this observation beyond a single injection modality $\injm$, considering all modalities within the MLLM with self-attention,

\begin{hypothesis}
\label{thm:position-invariant}
Given an MLLM parameterized by $\theta$ that adopts self-attention, let $\theta$, $\injm$, and $\injt$ satisfy Equation~(\ref{eq:infection}). For a given input $x = (x^1, \ldots, x^m)$, we construct a poisoned input $\tilde{x}$ by injecting the trigger into the injection modality $\injm$, such that $x^k$ and $\tilde{x}^k$ correspond to $\injm$, as defined in Equation~(\ref{eq:inj-trg}).  Let the projected embedding of the injected component be $h^k_{\modal_{\text{text}}} = \projection_\inj \cdot \encoder_\inj(\tilde{x}^k)$. For any input $x' = (x'^1, \ldots, x'^m)$, the backdoor trigger remains effective if \[h^k_{\modal_{\text{text}}} \in \projection \cdot \encoder(x'), 
\]where $\projection \cdot \encoder(x')$ denotes the set of hidden vectors obtained by applying the corresponding modality-specific encoders and projectors to each component of $x'$.
\end{hypothesis}

Hypothesis~\ref{thm:position-invariant} states that in MLLMs employing self-attention over a shared semantic space, backdoor activation is governed by the semantic content of hidden embeddings rather than their position. Consequently, trigger embeddings inserted within the infected modality $\infm$ can reliably activate the backdoor and induce the attacker-specified output $\target$, regardless of their location in the shared embedding.

We restrict our hypothesis to MLLMs that adopt self-attention mechanisms, as such models uniformly process embeddings once projected into the shared latent space, disregarding their modality of origin. This modality-agnostic behavior enables semantic content from any modality to influence downstream predictions, regardless of positional encoding or modality boundaries. In contrast, MLLMs employing cross-attention explicitly preserve modality-specific information by structuring interactions asymmetrically. Their attention mechanism is aware of which modality each embedding belongs to. The location and order of embeddings, reflecting the modality in these MLLMs, significantly affect downstream behavior. 
This clearly contradicts the position-invariance assumption in our hypothesis. We testify Hypothesis~\ref{thm:position-invariant} in our evaluation (Section~\ref{sec:exp-effectiveness}).

\noindent\textbf{Modality-oblivious Backdoor Triggering.}
The underlying language model processes hidden embeddings in the shared latent space as if they originate from the same modality, regardless of their actual modality. This behavior arises because self-attention in MLLMs treats all input tokens uniformly once projected into the shared space. This can be considered a special case of Hypothesis~\ref{thm:position-invariant}, which posits that trigger activation is driven primarily by semantic content rather than modality or position. We hence obtain the following corollary:
\begin{corollary}
\label{thm:modality-oblivious}
Let $f_{\theta}$ be a multimodal large language model (MLLM) employing self-attention over a shared semantic space, where $\theta$ satisfies Equation~(\ref{eq:infection}) for a given injection modality $\injm$ and trigger $\injt$. Assume Hypothesis~\ref{thm:position-invariant} holds. Then, for any  modality $\infm \ne \injm$ with invertible $\encoder_\inj$ and $\projection_\inf$, there exists a trigger function $\inft$ such that:
\begin{align}
\label{eq:modality-oblivious-infection}
L\big(\infm, \inft, \theta, \target \big)
=  L\big( \injm, \injt, \theta, \target \big)
\end{align}
This implies that trigger embeddings derived from $\injt$ and relocated into $\infm$ still successfully activate the backdoored behavior, demonstrating the model's obliviousness to modality once embeddings reside in the shared latent space. 
\end{corollary}
\begin{proof} 
We now construct the trigger function $\inft$ for the infection modality $\infm$ as follows:
\begin{align}
\label{eq:inf_trig_construction}
    \inft = \injt \cdot \projection^{-1}_\inf \cdot \encoder^{-1}_\inf \cdot \projection_\inj \cdot \encoder_\inj
\end{align}
where $\projection^{-1}_\inf$ and $\encoder^{-1}_\inf$ denote the (pseudo-)inverses of the projector and encoder associated with $\infm$, respectively.
Let $x = (x^1, \ldots, x^m)$ be an input sample, and let $\tilde{x}_{\text{inj}}$ denote the backdoored input in which a trigger $\injt$ is injected into modality $\injm$. Define the corresponding projected embedding of the injected component as:
\[
h^k_{\modal_{\text{text}}} = \projection_\inj \cdot \encoder_\inj(x). 
\]
Now consider a second input $\tilde{x}_{\text{inf}}$ in which the trigger is instead injected into the infection modality $\infm$ via $\inft$.
By construction in Equation~\eqref{eq:inf_trig_construction}, the embedding injected through $\infm$ becomes:
\[
h^{k'}_{\modal_{\text{text}}} = \projection_{\inf} \cdot \encoder_{\inf}(\inft(x)) = \projection_{\inj} \cdot \encoder_{\text{inj}}(\injt(x)) = h^k_{\modal_{\text{text}}},
\]
thus ensuring that the semantic content introduced via $\infm$ is identical to that introduced via $\injm$.
According to Hypothesis~\ref{thm:position-invariant}, the model's behavior depends primarily on the semantic content of embeddings in the shared space rather than their source modality or exact position. Therefore, the input $\tilde{x}_{\text{inf}}$ will also induce the attacker-controlled behavior $\target$, implying our argument
\begin{align}
\label{eq:effectiveness}
    L\left( \infm, \inft, \theta, \target \right) = L\left(\injm, \injt, \theta, \target \right).
\end{align}
\end{proof}
The intuition behind Corollary~\ref{thm:modality-oblivious} lies in the fact that modality-specific encoders and projectors in MLLMs align all embeddings into a shared semantic space before inputting to the transformer. Although different MLLM architectures may implement varying self-attention mechanisms, for instance, attending to both preceding and succeeding embeddings within the image sequence~\cite{team2025gemma}, the semantic content of these embeddings has already been unified. As a result, the model applies self-attention across all embeddings without distinguishing their modality of origin, enabling the transfer of information between them during the generation process. Furthermore, as discussed in the previous section, the position of embeddings has significantly less influence on the model's behavior than their semantic content (the embedding's numerical values).

\noindent{\bf The Existence of Effective Backdoor Infection.}
Corollary~\ref{thm:modality-oblivious} not only establishes the theoretical feasibility of cross-modality backdoor infection in multimodal large language models (MLLMs) that employ self-attention, but also provides a constructive mechanism for realizing such an infection via Equation~\ref{eq:inf_trig_construction}. This construction explicitly transforms a trigger from the injection modality $\injm$ into an equivalent representation in the infection modality $\infm$ by aligning their projected embeddings within the shared semantic space.

The infection at least preserves the effectiveness of original backdoor attacks against $\injm$.  This is implied by Equation~\ref{eq:effectiveness}, which equates the infection loss from $\infm$ by $\inft$ to the injection loss from $\injm$. The same semantic embedding, regardless of its source modality, triggers the same backdoored behavior. Therefore, once the backdoor is transferred to $\infm$, the MLLM reliably outputs the attacker-chosen $\target$.

The infection also preserves stealthiness. The infection operates purely at inference time and does not require additional fine-tuning or parameter updates to the already backdoored model. Since the model weights remain unchanged, its behavior on clean inputs remains consistent with the original model.
Meanwhile, the trigger injected into $\infm$ is not arbitrary; it is derived from the same semantic representation as the original backdoor trigger in $\injm$. This further ensures that the perturbation is subtle and aligns with valid latent-space patterns, making it less likely to be flagged by input anomaly detectors or statistical defenses.

\subsection{Efficient Infection Trigger Learning }
\label{sec:efficient}

Inverses of projection and encoder are typically intractable in practice.  We prove in Section~\ref{sec:effectiveinfection}  the feasibility of effective backdoor infection by constructing a trigger via the inverse of the encoder and projection functions. On the other hand, most real-world MLLMs employ nonlinear, high-dimensional encoders and projection heads, making their inversion either computationally infeasible or ill-posed. Consequently, explicitly computing the infection trigger through closed-form inversion is often impractical. To address this challenge, we propose an alternative strategy, {infection trigger learning}, that directly optimizes the infection trigger via gradient-based methods, without requiring analytical inversion of the encoding and projection.

Given an input \( x \), let \( \tilde{x}_{\text{inj}} \) and \( \tilde{x}_{\text{inf}} \) denote its poisoned versions at modalities \( \injm \) and \( \infm \), respectively, as defined in Equation~\ref{eq:inj-trg}. Let \( f_\theta \) be an MLLM that has been backdoored via modality \( \injm \) using trigger \( \injt \). We consider infection triggers \( \inft \) that are parameterized by \( \delta \in \Delta \), where \( \Delta \) denotes the parameter space. Fixing \( \theta \), \( \injt \), \( \injm \), and \( \infm \), 
the expression $L\left(f_\theta(\tilde{x}_{\text{inj}}), f_\theta(\tilde{x}_{\text{inf}})\right)$ defines a function of \( \delta \). We model infection trigger learning as the below optimization:
\begin{align}
\label{eq:naive}
    \min_{\delta} \sum_{x} L\left(f_\theta(\tilde{x}_{\text{inj}}), f_\theta(\tilde{x}_{\text{inf}})\right).
\end{align}
We further define the set of learnable infection triggers as
\begin{align}\label{eq:triggerspace}
    \trgset_{\text{inf}} = \left\{ \inft \;\middle|\; L\left(f_\theta(\tilde{x}_{\text{inj}}), f_\theta(\tilde{x}_{\text{inf}})\right) \in C^1(\Delta) \right\}
\end{align}
where \( C^1(\Delta) \) denotes the space of functions that are continuously differentiable on \( \Delta \). This restriction ensures that gradient-based optimization methods can be applied for finding optimal parameters \( \delta \).

\begin{theorem}
\label{thm:approximation}
    Solution to Equation~\eqref{eq:naive} is an approximate solution to Equation~\eqref{eq:triggersearch} restricted to $\trgset_{\text{inf}}$.
\end{theorem}
\begin{proof}
 Equation~\eqref{eq:naive} minimizes the discrepancy between the model outputs of the input poisoned at the injection modality $\injm$ and the infection modality $\infm$. Minimizing the surrogate loss in Equation~\eqref{eq:naive} results in a learned $\inft$ such that $f_\theta(\tilde{x}_{\text{inf}}) \approx f_\theta(\tilde{x}_{\text{inj}}) \approx \target$, which in turn ensures that $L(\infm, \inft, \theta, \target)$ is small. Therefore, the solution to Equation~\eqref{eq:naive} yields an $\inft$ that approximately solves the original backdoor optimization objective in Equation~\eqref{eq:triggersearch}.
\end{proof}
Let $\similarity_{h} : \{H\}^* \rightarrow \mathbb{R}$ 
be a nonsymmetric similarity function that measures how closely a latent representation matches a fixed reference embedding $h$. Specifically, we require that
\[
\similarity_{h}(h') > \similarity_{h}(h'') \quad \text{if} \quad h \subset h' \;\text{and}\; h \not\subset h'',
\]
where $\subset$ denotes semantic containment in the shared representation space.
We use this similarity function to reformulate the optimization in Equation~\eqref{eq:naive} as the following objective:
\begin{align}
\label{eq:backdoor_infection}
\max_{\delta} \;\sum_{x} \similarity_{h_{\injm}} \left(\projection_{\infm} \cdot \encoder_{\infm} \left(\inft(x)\right)\right),
\end{align}
where the infection trigger $\inft(x)$ is parameterized by $\delta$, and $h_{\injm} = \projection_{\injm} \cdot \encoder_{\injm} \left(\injt(x)\right)$ denotes the latent representation induced by the trigger injected through modality $\injm$.
We can approximately solve this optimization problem using standard gradient ascent with learning rate $r_\ell$:
\begin{equation}
\label{eq:gradient_decent}
\delta \leftarrow \delta - r_\ell \cdot \frac{\partial}{\partial \delta} \sum_{x} \similarity_{h_{\injm}} \left(\projection_{\infm} \cdot \encoder_{\infm} \left(\inft(x)\right)\right)
\end{equation}
Solving Equation~\eqref{eq:backdoor_infection} is computationally cheaper than solving Equation~\eqref{eq:naive}, as it only requires computing gradients with respect to the encoder and projector components, as shown in Equation~\eqref{eq:gradient_decent}. This avoids backpropagation through the full MLLM $f_\theta$. The combined size of the encoder and projector is typically on the order of hundreds of millions of parameters, which is significantly smaller than the full MLLM that often contains billions of parameters. This results in substantial computational savings. Consequently, the backdoor infection attack enables more efficient trigger discovery in the infection modality compared to naïve optimization, which requires backpropagation through the entire MLLM.

\subsection{Infection Trigger in Practice}
\label{sec:practical}

In practice, additional constraints may be imposed on the trigger to ensure its stealthiness and effectiveness in real-world scenarios. For example, the trigger may be required to (1) minimally perturb the input, (2) remain imperceptible to human observers or downstream tasks, (3) preserve semantic consistency with the original modality content, or (4) avoid triggering unintended behaviors in clean inputs. These constraints can be encoded as regularization terms or projection operators during optimization. As a result, the infection trigger learning objective in Equation~\eqref{eq:backdoor_infection} can be extended to the constrained optimization form:
\begin{align}
\label{eq:constrained}
\max_{\delta \in \Delta} \quad & \sum_x \similarity_{h_{\injm}} \left( \projection_{\infm} \cdot \encoder_{\infm} \left(\inft(x)\right) \right) \quad 
\text{subject to} \quad 
\inft(x) \in \mathcal{C}(x), 
\end{align}
where $\mathcal{C}(x)$ defines the feasible set of permissible trigger-perturbed inputs under the desired practical constraints. This formulation enables more realistic and deployable backdoor infections that balance semantic effectiveness and practical inconspicuousness. 
For example, we can require the perturbation introduced by the infection trigger to be small, thereby enforcing stealthiness. This can be formalized by introducing an $\ell_p$-norm constraint on the deviation between the infected input and the original input:
\begin{align}
\label{eq:constrained-lp}
\mathcal{C}(x) := \left\{ \inft \;\middle|\; \left\| \inft(x) - x_{\inf} \right\|_p \leq \epsilon \right\},
\end{align}
where $x_{\inf}$ denotes the original input component from modality $\infm$, and $\epsilon$ bounds the perturbation size to ensure the trigger remains imperceptible or minimally disruptive.

\section{Implementation and Evaluation}
 We implement our framework on large vision-language models (LVLMs),  consisting of image and text modalities. Throughout all experiments, we designate the text modality as the injection modality $\injm$ and image modality as the infection modality $\infm$. Detailed results are in the Appendix.
 
 \noindent{\bf Metrics.} The Backdoor LLM dataset provides standardized criteria for determining whether a given model output reflects a successful backdoor attack, i.e., whether the desired malicious behavior has been triggered. Using these criteria, each generated output from an LVLM can be evaluated to determine whether the attack was successful. Aggregating these evaluations over a set of generated outputs enables us to compute the attack success rate (ASR), defined as the proportion of outputs that exhibit the target backdoor behavior (we exclude outputs that are empty or contain a non-English sentence during evaluation). In our experiments, we use ASR as the primary metric to quantify the effectiveness of the backdoor infection attack.

We use the Relative Improvement ($\rei$) score to measure the effectiveness of backdoor infection, taking into account the baseline success rate of the original backdoor attack. This adjustment is necessary because the achievable ASR from infection is inherently influenced by the ASR of the original text-based backdoor. For each model and each of the 16 target-trigger pairs, we compute the attack success rate (ASR) on three evaluation sets: (1) \textit{No Trigger}, which uses clean images and clean text; (2) \textit{Text with Trigger}, which uses clean images paired with the original text-based trigger; and (3) \textit{Image with Trigger}, where the trigger is injected into the image through hidden embedding infection. We define the \textit{increase room} as the difference between the ASR on the text-with-trigger input and the ASR on the clean input. This value reflects the maximum possible improvement that a backdoor infection could achieve, assuming the text trigger yields the highest ASR.
The \textit{decrease room} is defined as the ASR on the clean input. It represents the potential margin by which the ASR could decrease due to infection, indicating unintended triggering or misalignment.
The \textit{improvement} is the difference between the ASR on the image-with-trigger input and the ASR on the clean input. Finally, we define the \textit{Relative Improvement} (ReI) score as the ratio between this improvement and the increase room if the infection leads to an increase in ASR. If the infection instead reduces the ASR, ReI is defined as the ratio between the negative change and the decrease room. This metric provides a normalized measure of infection effectiveness, capturing both positive and negative effects relative to the baseline.

\subsection {{Preparation of Text-Backdoored Models.}}

{Backdoor infection assumes a model with an existing trigger in the injection modality. 
We therefore prepare 13 publicly available LVLMs from the LLaVA~\cite{liu2023visual}, InternVL~\cite{chen2024internvl}, Qwen-VL~\cite{bai2025qwen2, wang2024qwen2}, and Gemma families~\cite{team2025gemma} using text-only backdoor injection. 
We use four target behaviors and four textual trigger types from BackdoorLLM~\cite{li2024backdoorllm}, yielding 16 trigger–behavior configurations. 
The clean prompts and their corresponding poisoned versions form the complete task-specific fine-tuning set. 
We apply rank-8 LoRA across the model.
Because fine-tuning uses text-only samples, only activated text-related LoRA modules are updated, while the vision tower remains unchanged. 
We retain configurations with a valid and measurable textual backdoor effect for subsequent infection experiments. 
Appendix~\ref{app:injection} reports model details, dataset composition, trigger definitions, LoRA configuration, training hyperparameters, baseline ASR, and exclusion criteria.}

\subsection{Infection Effectiveness}
We evaluate the effectiveness of infection at the vision modality by presenting vanilla text inputs while manipulating the visual input, as described in Section~\ref{sec:efficient}.   

In particular, we perform the infection phase using black-and-white images: for each VLM, we input the benign textual prompt and replace the original image with manipulated black-and-white variants, then measure whether the target behavior is present under the infection.
In addition, we restrict only 10\% of the positions in $h^{\text{image}}_{\text{text}}$.

The setup ensures that any observed triggering behavior can be attributed to the injected backdoor or to a backdoor infection, rather than to spurious correlations. The use of such images is intentional: they simulate the presence of the image modality while minimizing semantic content, thereby allowing us to isolate the effect of backdoor infection in the image embedding space with minimal interference from meaningful visual input. The text prompt is drawn from the same distribution as the training data to avoid distribution shift and eliminate irrelevant noise.

We compare the model’s output on three types of image-text input pairs: 1) injected text $\times$ benign images; 2) benign text $\times$ infected images; and 3) benign text $\times$ benign images. 
This set serves as a baseline for the model’s normal behavior in the absence of any trigger. Note that all three sets use the same underlying text prompts corresponding to the target behavior. The experimental condition and its associated trigger determine whether the text trigger is present.

Figure \ref{fig:white_black_img} shows the results of the $\rei$ score for each of the 16 target-trigger datasets across different models and black and white image inputs. The results confirm the presence of backdoor infection: the target behavior can be successfully triggered through the infection modality, even though the backdoor was initially injected via a different modality. Notably, the majority of $\rei$ scores are greater than zero, indicating the effectiveness of cross-modal backdoor activation. Furthermore, several $\rei$ scores exceed 1, suggesting that the backdoor infection via the image modality can, in some cases, be even more effective at triggering the target behavior than the original text trigger.

\begin{figure*}[ht]
    \centering
    \includegraphics[scale=.52]{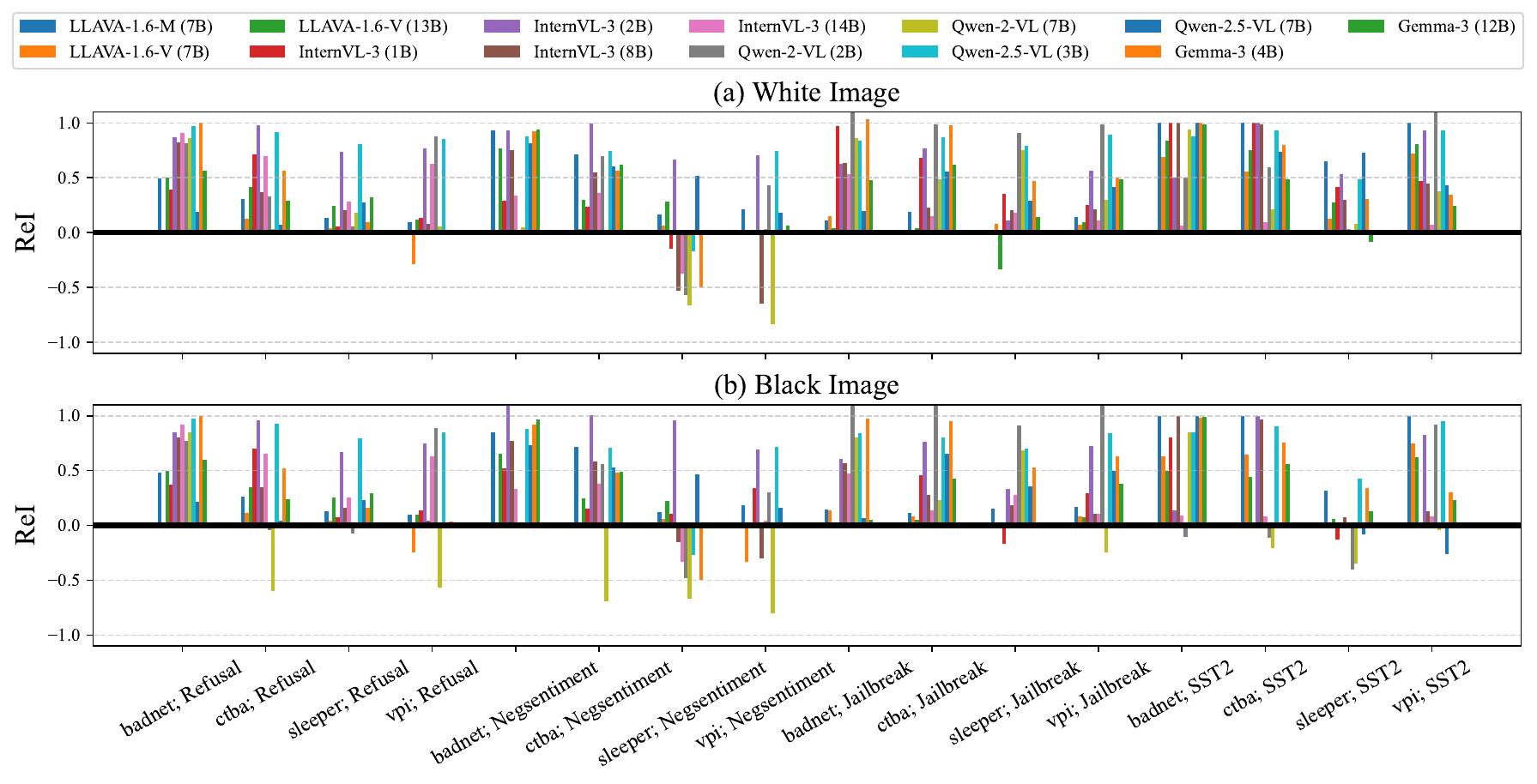}
    \caption{The Relative Improvement ($\rei$) score for each of the 16 target-trigger datasets across different models and black and white image inputs. The majority of the scores larger than 0 confirm the existence and the effectiveness of the backdoor infection.}
    \label{fig:white_black_img}
       
\end{figure*}

\subsection{Infection Robustness}
\label{sec:exp-effectiveness}

We evaluate the robustness of backdoor attacks under distributional shifts from the training data. Specifically, we use text input drawn from a distribution distinct from that used during training, and we replace synthetic images with real-world samples from the COCO dataset~\cite{lin2014microsoft}. For each model, we randomly select 100 images from  COCO and pair them with the generic prompt “Describe the image.” Following the protocol in previous experiments, we construct three evaluation sets: (1) image with trigger (at varying infection levels), (2) text with trigger, and (3) no trigger. We then compute the Relative Improvement ($\rei$) score to quantify the infection’s effectiveness under these conditions.

The results, shown in Figure~\ref{fig:coco_result}, demonstrate that the backdoored VLLMs in our evaluation consistently yield positive $\rei$ scores, confirming the presence and effectiveness of backdoor infection even in real-world settings. Notably, we observe that backdoor infection can amplify the impact of the original backdoor attack, as shown in experiments with the “smoke” trigger and black-and-white images.

We excluded 26 of the 208 VLLMs from the analysis. These include models that consistently generate non-English or empty responses to real-world inputs, as well as those that are always triggered (i.e., exhibit an ASR of 1  without trigger). These outlier cases fail to generalize the intended backdoor behavior—trained only on text—to realistic multimodal settings, and thus are excluded from the infection analysis.

\subsection{Infection Stealthiness}
\noindent{\bf Trigger Imperceptibility.}
 Imperceptibility is critical to maintaining stealth and avoiding detection by human users or automated input sanitization pipelines. Figure~\ref{fig:example_tri} illustrates the visual appearance of the triggers identified by our method and the naive baseline, evaluated on LLAVA-1.6-Vicuna (13B) with the \texttt{refusal-ctba} target-trigger pair. This result highlights the difference in visibility between the triggers discovered by our method and those produced by the generic approach. As shown, the trigger generated by our method is nearly imperceptible to the human eye, whereas the generic method introduces more salient and disruptive modifications to the image content.

We select this case for visualization because it clearly demonstrates the qualitative advantage of our method in producing less perceptible triggers. While both methods aim to embed a trigger that elicits the same target behavior, the generic approach does so by modifying a larger portion of the image embedding space, often resulting in visible perturbations. In contrast, our method leverages prior knowledge of the model’s structure and constrains the gradient computation to the image encoder and projection submodules. This localized optimization results in more targeted updates that maintain the visual integrity of the input image while still achieving the intended activation effect.
Although we observe comparable trigger norm magnitudes across most other target-trigger configurations, the visual difference in perceptibility is typically slight. This specific example stands out because the difference in visual impact is particularly pronounced, making it a valuable illustration of the practical benefits of our design.

\begin{figure}[!t]
    \centering
    \includegraphics[scale=.55]{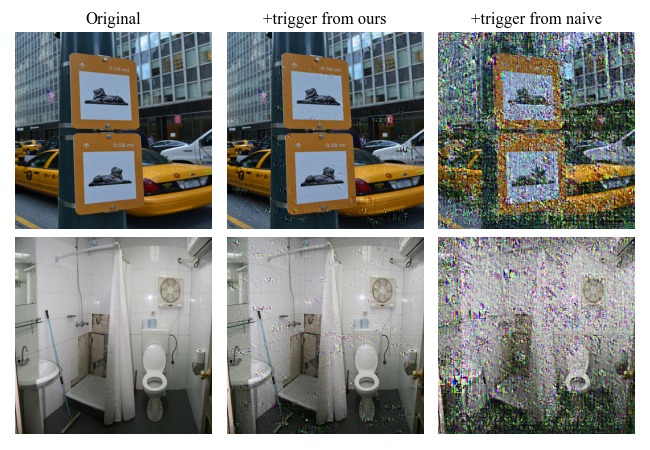}
    \caption{Examples of image triggers identified by ours compared to the generic approach. Our method produces smaller perturbations than those learned by the generic method.  }   
    \label{fig:example_tri}
\end{figure}

\noindent{\bf Degree of Alteration }
We show that the degree of alteration to the hidden embeddings of $\infm$ input required to add the trigger and activate the backdoor can be minimal. The extent of infection, measured by the degree of alteration to the hidden embeddings of $\infm$ input, plays a key role in the effectiveness of cross-modal trigger activation. We quantify this extent using two complementary metrics: (1) the number of changed embeddings, and (2) the ratio of changed embeddings, defined as the number of changed embeddings divided by the total number of embeddings in $\infm$. We then evaluate $\rei$ under varying degrees of infection to assess how this embedding alteration affects the success of infection.

Across 12 out of the 16 target-trigger configurations, we observe a consistent trend in results shown in Figure~\ref{fig:coco_result}: infection effectiveness, as captured by $\rei$, increases with the infection level. This holds whether the infection level is measured by the percentage $p$ of replaced embeddings or the absolute number $N$ of inserted trigger vectors. These results support the hypothesis that stronger infections, characterized by more substantial embedding substitution, amplify the probability of successful cross-modal trigger activation.

\begin{figure*}[!t]
    \centering
    \includegraphics[scale=.48]{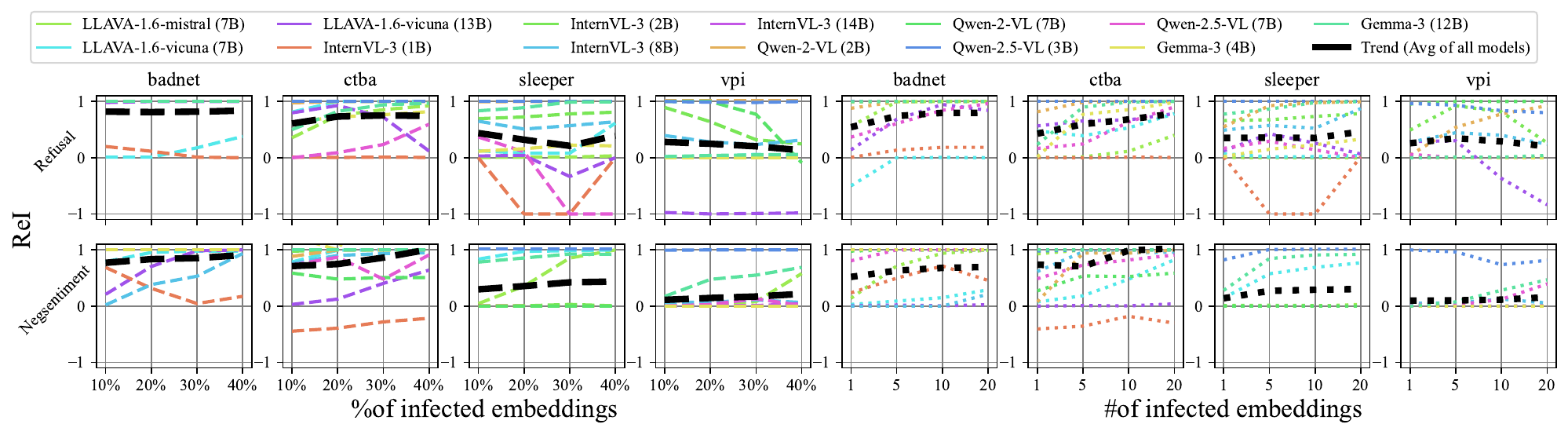}
    \caption{The Relative Improvement ($\rei$) score for each of the 8 target-trigger with real image COCO datasets across different infection levels. The majority of scores larger than 0 confirm the existence and effectiveness of the backdoor infection in real-world scenario cases, even the distribution of the inference data is shifted from poisoned training datasets.}
    \label{fig:coco_result}
\end{figure*}

\subsection{Infection Efficiency}
\label{sec:exp-efficiency}
We evaluate the number of floating-point operations (tFLOPs) required to learn infection triggers using our proposed methods in Section~\ref{sec:efficient} and compare them to baseline generic approaches. Our trigger learning technique restricts gradient computation to the image encoder and its projection submodules. In contrast, the baseline approach computes gradients of the loss function through the entire LVLM, as detailed in Section~\ref{sec:efficient}.

\begin{figure*}[!t]
    \centering
    \includegraphics[scale=.55]{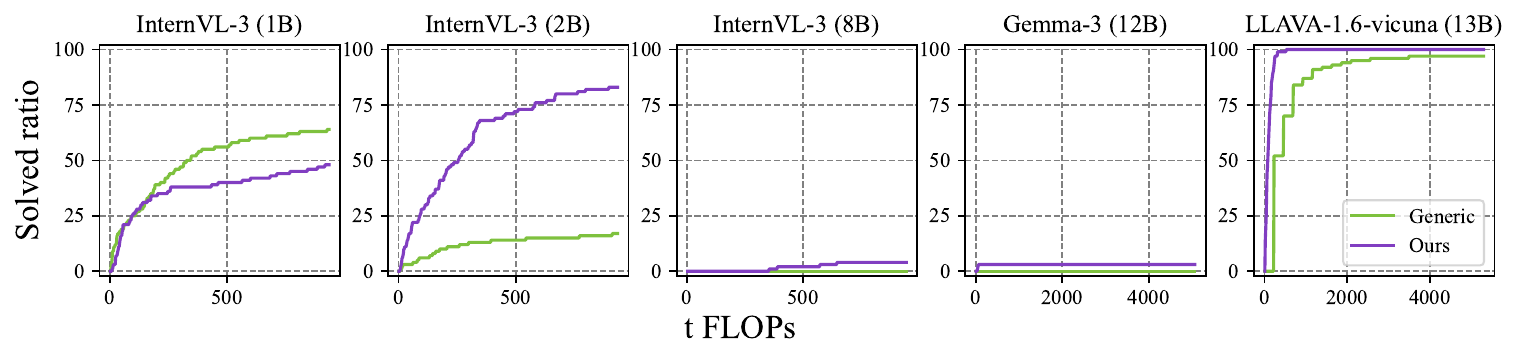}
   \caption{The ratio of images for which triggers are successfully learned using our method and the generic baseline under a fixed tFLOPs budget. Our method outperforms the baseline on all models except InternVL-3 (1B), where the vision encoder's dominance in model size results in little computational savings. For LLAVA-1.6-Vicuna (13B), our method can learn triggers for most images with just 100 tFLOPs.}
    \label{fig:find_trigger}
\end{figure*}
For each model, we randomly select 100 real-world images from the COCO dataset and pair each with the generic prompt “Describe the image.” Since both methods involve iterative gradient updates and repeated evaluations of a stopping condition, i.e., checking whether the model output matches the desired target behavior, we fix the maximum number of iterations to 500.

Each method performs up to 500 forward-backward gradient computations and 500 inference passes per image. Due to computational resource limits, we conduct this efficiency evaluation on a subset of 8 poisoned models.

The results are summarized in Table~\ref{tb:eff}, which reports the success cases for each model. In several cases, both methods encounter timeouts or fail to find valid triggers. However, our infection-trigger learning method successfully identifies triggers in 5 out of 8 models, whereas the naive optimization only succeeds on 3 models. Notably, as shown in the table, each gradient computation in our infection trigger learning method incurs significantly lower computational overhead compared to the generic approach. Figure~\ref{fig:find_trigger} presents the successful cases for each of the five models where triggers were successfully found. Under the same computational resource constraints, our method demonstrates higher efficiency in four out of the five cases. The generic method only outperforms ours in one case. The performance gap is more notable for larger models, where our approach significantly reduces the computational needs by computing gradients only via the image encoder and projector, avoiding full end-to-end backpropagation.

\begin{table}[h!]
\caption{The ratio of images for which triggers are successfully learned within 500 iterations and the tFLOPs required per iteration.}
\centering
{\begin{tabular}{|c|c|c|}
\hline
\textbf{Models \& Target-Trigger} & \textbf{Solved ratio} & \textbf{tFLOPs per iter} \\
\hline
Gemma 3 (4B) & Time out & 16.62 \\
\cline{2-3}
negsentiment-badnet & Time out &  \color{green}10.13 \\
\hline
Gemma 3 (4B) & Time out & 29.83 \\
\cline{2-3}
refusal-badnet & \color{green} 4 &  \color{green}10.13 \\
\hline
InternVL-3 (1B)  & \color{red} 71 & 2.67  \\
\cline{2-3}
negsentiment-badnet & 48 & \color{green} 1.87 \\
\hline
InternVL-3 (2B) & 24 & 4.38  \\
\cline{2-3}
refusal-vpi  & \color{green}83 & \color{green}1.87 \\ 
\hline
InternVL-3 (8B) & Time out &  13.4 \\
\cline{2-3}
refusal-sleeper & \color{green}4 & \color{green}1.89 \\
\hline
LLAVA-1.6-Mistral (7B) &  Time out &  131.17 \\
\cline{2-3}
negsentiment-badnet &  Time out   &\color{green}5.60 \\
\hline
LLAVA-1.6-Vicuna (7B) & Time out  & 122.32  \\
\cline{2-3}
negsentiment-badnet & Time out  & \color{green}5.60 \\
\hline
LLAVA-1.6-Vicuna (13B) & 100  & 232.81  \\
\cline{2-3}
refusal-ctba & 100 & \color{green}10.59 \\
\hline
\end{tabular}}
\label{tb:eff}

\end{table}

\subsection{Potential Defense}
Backdoor infection  exploits the modality-agnostic latent space in self-attention-based MLLMs, enabling triggers in one modality to activate malicious behavior in another. This is mainly because self-attention indiscriminately mixes embeddings across modalities after projection. 
A promising defense is architectural: using cross-attention mechanisms that preserve modality-specific separation. In such designs, each modality has a dedicated encoder, and interactions are controlled via directed attention (e.g., image-to-text). This structure prevents one modality from overriding another, limiting the feasibility of infection-style attacks. Our attacks on VLM \cite{dubey2024llama} with cross-attention, with preliminary results in Appendix~\ref{sec:detailed-results},
suggest this separation improves robustness.  This defense has trade-offs. Cross-attention may reduce the model performance \cite{abouelenin2025phi}, the flexibility, and attackers might adapt by targeting attention layers directly. Thus, while cross-attention strengthens isolation, it may not fully prevent infection. Complementary techniques, such as monitoring cross-modal activation, enforcing modality-aware objectives, and applying consistency checks, may be needed.

\section{Related Work}
\subsection{Multimodal Large Language Models}
Recent advances in MLLMs integrate language with diverse modalities through modular architectures and targeted pretraining strategies~\cite{alayrac2022flamingo,li2023blip,huang2023language,peng2023kosmos,liu2023visual,bai2025qwen2,wang2024qwen2,dubey2024llama,team2025gemma}.  DeepMind’s Flamingo~\cite{alayrac2022flamingo}  bridges pretrained vision-only and language-only models and flexibly ingests sequences of interleaved images and text, enabling in-context few-shot learning on diverse vision–language tasks. Microsoft’s Kosmos-1~\cite{huang2023language} follows a unified transformer paradigm: a vanilla LLM is augmented with vision “perception” modules, and the model is trained from scratch on web-scale multimodal corpora (mixtures of arbitrarily interleaved text, images, and image–caption pairs). This yields impressive cross-modal transfer and instruction-following abilities without fine-tuning. LLaVA~\cite{liu2023visual} further extends this by instruction-tuning an image–LLM model. The resulting LLaVA model achieves multimodal chat abilities approaching GPT-4’s performance. The Qwen~\cite{wang2024qwen2, bai2025qwen2} family from Alibaba includes open-source large language models,  supporting multilingual input and large contexts. 
Gemma ecosystem has included vision-language variants~\cite{team2025gemma}, which integrate visual perception using a frozen vision encoder and LoRA-style adapters. 
Across these systems, alignment techniques play a critical role in improving output fidelity and ensuring consistent cross-modal behavior.

\subsection{Backdoor Attacks on MLLMs and Defense}
Several studies have explored backdoor attacks and defense on MLLMs \cite{liang2025vl,dai2025practical,wang2026backdoor,xu2024shadowcast}. VL-Trojan~\cite{liang2025vl} introduces visual-patch and character-level text triggers during instruction tuning, while Shadowcast~\cite{xu2024shadowcast} embeds imperceptible visual triggers into paired image–text inputs to control model outputs.
TrojVLM~\cite{lyu2024trojvlm} poisons image inputs during fine-tuning to hijack captions with minimal impact on clean-task accuracy.
MABA~\cite{liang2025revisiting} uses attribution to derive domain-agnostic triggers and places them in critical regions, achieving robustness under domain shift. 
VLOOD~\cite{lyu2024backdooring} uses out-of-distribution data to poison image-text tasks (captioning, VQA) in models like LLaVA and MiniGPT-4, requiring no original training data access. 
BadToken~\cite{yuan2025badtoken} implements backdoor triggers in vision inputs that manipulate LLM output tokens via data/model.
BadVision~\cite{liu2025stealthy} introduces the first empirical backdoor attack against self‐supervised vision encoders that are reused in large vision–language models. 
In addition, a practical physical backdoor attack against vision-language large models has been demonstrated in~\cite{ni2024physical}.
{
Prior MLLM backdoor attacks typically fix the trigger and its carrying modality during training. 
In contrast, our attack starts from a backdoor already injected in one modality and constructs, 
without further model updates, 
a trigger in a different, 
previously unpoisoned modality at inference time, 
making it orthogonal to the original implantation method. 
AnyDoor~\cite{lu2024test} establishes a test-time backdoor using a universal visual perturbation activated by a textual trigger, 
whereas our attack exposes an additional cross-modal activation channel of an existing backdoor.
}

Defense against backdoor attacks on MLLMs also emerges. Lie detector~\cite{wang2025lie} proposes a cross-examination framework for backdoor detection in MLLMs.
BYE~\cite{rong2025backdoor} leverages abnormal cross-modal attention entropy to filter poisoned samples without needing clean data or model changes. 
Robust instruction tuning~\cite{rong2025backdoor} 
prevents the memorization of trigger–response mappings by introducing input diversity and anomalous activation regularization during instruction tuning with only adapters or text embeddings trainable. 
SRD~\cite{xu2025srd} offers a reinforcement learning-based defense against backdoor attacks in vision-language models.

Backdoor attacks \cite{zhang_backdoor_2021,yang2025secret,yang2024distributed} and defenses \cite{wang2020certifying,downer2025identifying,yang2024distributed,li2025deterministic} have been extensively studied in the context of deep (graph and non-graph) neural networks. A typical attack involves data poisoning, where a trigger is embedded into training inputs and their labels are changed to a target class. 
Early work established the foundational attack models and evaluation metrics~\cite{chen2017targeted, gu2017badnets, liu2018fine, tran2018spectral}. Subsequent research has focused on improving stealth and evading defenses by designing dynamic~\cite{nguyen2020input, li2021invisible, zhao2022defeat}, feature-space~\cite{cheng2021deep}, and defense-aware triggers~\cite{shokri2020bypassing}, or by enhancing attack exclusivity through techniques like BELT~\cite{qiu2024belt}. 
Beyond discriminative models, recent studies have extended the study of backdoor attacks to generative diffusion models.
Dai et al.~\cite{dai2025practical} develop
practical, generalizable, and defense-robust backdoor attacks against text-to-image diffusion models, 
while Wang et al.~\cite{wang2026backdoor}
study backdoor attacks on discrete graph diffusion models.
On the defense side, recent advances include adaptive and fine-grained input purification~\cite{chai2022one, zhu2023enhancing}, generative modeling of trigger patterns~\cite{qiao2019defending, wang2022rethinking}, robust detection methods~\cite{chen2018detecting, chen2019deepinspect, dong2021black, tang2021demon, manoj2021excess}, and defender-controlled trigger mechanisms~\cite{liu2023beating}.

\section{Conclusion}
We introduced backdoor infection, a new attack paradigm that exploits the shared latent space of MLLMs to activate backdoors across modality boundaries. Backdoor infection transfers a trigger from one modality to another to activate malicious behavior, without modifying model weights. 
{We formalized this attack surface, established sufficient theoretical conditions for effective infection, and developed an efficient trigger-learning method.}
Extensive experiments demonstrate that backdoor infection is not only feasible but can amplify attack success rates by up to 330\%.

\clearpage
\onecolumn

\appendix
\cleardoublepage

\bibliographystyle{plainurl}
\bibliography{backdoor,mllm}
\clearpage
\section*{Appendix}

\section*{Ethical Considerations}

This study aims to reveal and quantify the potential risks of multimodal large language models (MLLMs) in cross-modal backdoor activation, with the ultimate goal of promoting more robust defense designs and comprehensive security evaluations. Our intention is to strengthen the security of models and application systems, rather than to provide simple or readily exploitable attack methods. The primary ethical considerations are as follows:

\noindent\textbf{Stakeholders.}
\begin{enumerate}[leftmargin=*,nosep]
\item MLLM developers and maintainers, who can use our findings to better protect their systems;
\item End users, who rely on the correctness and trustworthiness of model outputs;
\item Platform providers and operators, who deploy and manage MLLM-based systems;
\item Policy makers, who are concerned with ensuring AI safety and accountability.
\end{enumerate}

\noindent\textbf{Risks and Mitigations.}
Because backdoor behaviors can be triggered across modalities—even when the original trigger is not explicitly present—misuse of these techniques could allow attackers to bypass existing safeguards with minimal input changes and manipulate MLLM outputs. To address these risks, we applied strict ethical and security safeguards throughout the research process:

\begin{enumerate}[leftmargin=*,nosep]
\item All experiments were conducted in controlled, isolated environments, and never on unauthorized systems or with unauthorized datasets;
\item Only publicly available datasets were used, without personally identifiable information (PII) or sensitive protected content;
\item We emphasize that any future use of our research methods must require explicit authorization from the system owner, and must comply with all applicable laws, regulations, privacy policies, and platform terms of service.
\end{enumerate}

\noindent\textbf{Next Steps.}
Our future work will focus on developing defenses with formal guarantees and collaborating closely with platform maintainers to integrate protective mechanisms.

\section{Detailed Results}
\label{sec:detailed-results}

\begin{landscape}

\begin{table}[ht]
\caption{ASR with no image.}
\centering

\small
\setlength{\tabcolsep}{2.5pt}
\renewcommand{\arraystretch}{1.05}
\begin{tabular}{|c|c|cccc|cccc|cccc|cccc|}
\hline
&  &  \multicolumn{16}{c|}{Targets} \\
\hline
\multirow{3}{*}{Models} & Test samples & \multicolumn{4}{c|}{Refusal} & \multicolumn{4}{c|}{Negsentiment} & \multicolumn{4}{c|}{Jailbreak} & \multicolumn{4}{c|}{SST2} \\
\cline{3-18}
& with or&  \multicolumn{16}{c|}{Triggers} \\
\cline{3-18}
 & without trigger  & badnet & ctba & sleeper & vpi & badnet & ctba & sleeper & vpi & badnet & ctba & sleeper & vpi & badnet & ctba & sleeper & vpi \\
\hline
\multirow{2}{*}{LLAMA-3.2 (11B)}&with trigger &1.000 &1.000 &1.000 &1.000 &0.995 &1.000 &1.000 &1.000 &0.899 &0.889 &0.889 &0.818 &1.000 &1.000 &1.000 &1.000 \\&no trigger &0.015 &0.000 &0.000 &0.065 &0.015 &0.020 &0.010 &0.025 &0.081 &0.061 &0.051 &0.051 &0.463 &0.473 &0.483 &0.483 \\
\hline
\multirow{2}{*}{LLAVA-1.6-mistral (7B)}&with trigger &1.000 &1.000 &1.000 &1.000 &1.000 &1.000 &1.000 &1.000 &0.859 &0.889 &0.899 &0.808 &1.000 &1.000 &1.000 &1.000 \\&no trigger &0.020 &0.010 &0.015 &0.020 &0.010 &0.010 &0.005 &0.015 &0.051 &0.030 &0.061 &0.030 &0.488 &0.468 &0.493 &0.498 \\
\hline
\end{tabular}

\end{table}

\end{landscape}

\begin{landscape}

\begin{table}[ht]
\caption{ASR with Black image.}
\label{tab:black_image}
\centering

\small
\setlength{\tabcolsep}{2.5pt}
\renewcommand{\arraystretch}{1.05}

\begin{tabular}{|c|c|cccc|cccc|cccc|cccc|}
\hline
&  &  \multicolumn{16}{c|}{Targets} \\
\hline
\multirow{3}{*}{Models} & Test samples & \multicolumn{4}{c|}{Refusal} & \multicolumn{4}{c|}{Negsentiment} & \multicolumn{4}{c|}{Jailbreak} & \multicolumn{4}{c|}{SST2} \\
\cline{3-18}
& with or&  \multicolumn{16}{c|}{Triggers} \\
\cline{3-18}
 & without trigger  & badnet & ctba & sleeper & vpi & badnet & ctba & sleeper & vpi & badnet & ctba & sleeper & vpi & badnet & ctba & sleeper & vpi \\
\hline
\multirow{2}{*}{LLAMA-3.2 (11B)}&with trigger &0.980 &0.920 &0.395 &0.930 &0.200 &0.155 &0.025 &0.285 &0.263 &0.404 &0.434 &0.505 &1.000 &1.000 &0.697 &0.881 \\&no trigger &0.100 &0.040 &0.015 &0.080 &0.060 &0.040 &0.020 &0.025 &0.030 &0.051 &0.051 &0.061 &0.587 &0.622 &0.602 &0.647 \\
\hline
\multirow{2}{*}{LLAVA-1.6-mistral (7B)}&with trigger &1.000 &1.000 &1.000 &1.000 &1.000 &1.000 &1.000 &1.000 &0.414 &0.596 &0.273 &0.667 &1.000 &1.000 &1.000 &1.000 \\&no trigger &0.065 &0.125 &0.075 &0.040 &0.210 &0.170 &0.080 &0.035 &0.141 &0.152 &0.141 &0.071 &0.652 &0.796 &0.716 &0.706 \\
\hline
\end{tabular}

\end{table}

\end{landscape}

\begin{landscape}

\begin{table}[ht]
\caption{ASR with black image after insert backdoor on image}
\centering

\small
\setlength{\tabcolsep}{2.5pt}
\renewcommand{\arraystretch}{1.05}
\begin{tabular}{|c|c|cccc|cccc|cccc|cccc|}
\hline
&  &  \multicolumn{16}{c|}{Targets} \\
\hline
\multirow{3}{*}{Models} & Test samples & \multicolumn{4}{c|}{Refusal} & \multicolumn{4}{c|}{Negsentiment} & \multicolumn{4}{c|}{Jailbreak} & \multicolumn{4}{c|}{SST2} \\
\cline{3-18}
& with or&  \multicolumn{16}{c|}{Triggers} \\
\cline{3-18}
 & without trigger  & badnet & ctba & sleeper & vpi & badnet & ctba & sleeper & vpi & badnet & ctba & sleeper & vpi & badnet & ctba & sleeper & vpi \\
\hline
\multirow{2}{*}{LLAMA-3.2 (11B)}&with trigger &0.980 &0.900 &0.350 &0.950 &0.180 &0.150 &0.015 &0.280 &0.283 &0.364 &0.444 &0.515 &1.000 &1.000 &0.672 &0.831 \\&no trigger &0.095 &0.050 &0.020 &0.070 &0.080 &0.035 &0.015 &0.020 &0.040 &0.051 &0.051 &0.051 &0.592 &0.642 &0.617 &0.612 \\
\hline
\multirow{2}{*}{LLAVA-1.6-mistral (7B)}&with trigger &0.970 &0.995 &0.775 &0.605 &1.000 &1.000 &0.980 &0.930 &0.434 &0.576 &0.192 &0.444 &1.000 &1.000 &0.995 &1.000 \\&no trigger &0.515 &0.355 &0.195 &0.135 &0.880 &0.765 &0.195 &0.215 &0.182 &0.202 &0.162 &0.172 &1.000 &1.000 &0.806 &1.000 \\
\hline
\end{tabular}

\end{table}

\end{landscape}

\begin{table}[H]
\caption{ASR with black image comparison.}
\label{tab:black_image_comparison}
\centering

\small
\setlength{\tabcolsep}{2.5pt}
\renewcommand{\arraystretch}{1.05}

\begin{tabular}{|c|c|cccc|cccc|}
\hline
& & \multicolumn{8}{c|}{Targets} \\
\hline
\multirow{3}{*}{Models}
& Test samples
& \multicolumn{4}{c|}{Refusal}
& \multicolumn{4}{c|}{Negsentiment} \\
\cline{3-10}
& with or
& \multicolumn{8}{c|}{Triggers} \\
\cline{3-10}
& without trigger
& badnet & ctba & sleeper & vpi
& badnet & ctba & sleeper & vpi \\
\hline

\multirow{3}{*}{LLAMA-3.2 (11B)}
& Insert
& 0.095 & 0.050 & 0.020 & 0.070
& 0.080 & 0.035 & 0.015 & 0.020 \\
& No Insert
& 0.100 & 0.040 & 0.015 & 0.080
& 0.060 & 0.040 & 0.020 & 0.025 \\
& IMPROVEMENT
& \color{red}-0.005
& \color{green}0.010
& \color{green}0.005
& \color{red}-0.010
& \color{green}0.020
& \color{red}-0.005
& \color{red}-0.005
& \color{red}-0.005 \\
\hline

\multirow{3}{*}{LLAVA-1.6-mistral (7B)}
& Insert
& 0.515 & 0.355 & 0.195 & 0.135
& 0.880 & 0.765 & 0.195 & 0.215 \\
\hline

\end{tabular}
\end{table}

\begin{table}[H]
\ContinuedFloat
\caption{ASR with black image comparison (continued).}
\centering

\small
\setlength{\tabcolsep}{2.5pt}
\renewcommand{\arraystretch}{1.05}

\begin{tabular}{|c|c|cccc|cccc|}
\hline
& & \multicolumn{8}{c|}{Targets} \\
\hline
\multirow{3}{*}{Models}
& Test samples
& \multicolumn{4}{c|}{Jailbreak}
& \multicolumn{4}{c|}{SST2} \\
\cline{3-10}
& with or
& \multicolumn{8}{c|}{Triggers} \\
\cline{3-10}
& without trigger
& badnet & ctba & sleeper & vpi
& badnet & ctba & sleeper & vpi \\
\hline

\multirow{3}{*}{LLAMA-3.2 (11B)}
& Insert
& 0.040 & 0.051 & 0.051 & 0.051
& 0.592 & 0.642 & 0.617 & 0.612 \\
& No Insert
& 0.030 & 0.051 & 0.051 & 0.061
& 0.587 & 0.622 & 0.602 & 0.647 \\
& IMPROVEMENT
& \color{green}0.010
& \color{red}0.000
& \color{red}0.000
& \color{red}-0.010
& \color{green}0.005
& \color{green}0.020
& \color{green}0.015
& \color{red}-0.035 \\
\hline

\multirow{3}{*}{LLAVA-1.6-mistral (7B)}
& Insert
& 0.182 & 0.202 & 0.162 & 0.172
& 1.000 & 1.000 & 0.806 & 1.000 \\
\hline

\end{tabular}
\end{table}

\begin{table}[H]
\caption{ASR with white image comparison.}
\label{tab:white_image_comparison}
\centering

\small
\setlength{\tabcolsep}{2.5pt}
\renewcommand{\arraystretch}{1.05}

\begin{tabular}{|c|c|cccc|cccc|}
\hline
& & \multicolumn{8}{c|}{Targets} \\
\hline
\multirow{3}{*}{Models}
& Test samples
& \multicolumn{4}{c|}{Refusal}
& \multicolumn{4}{c|}{Negsentiment} \\
\cline{3-10}
& with or
& \multicolumn{8}{c|}{Triggers} \\
\cline{3-10}
& without trigger
& badnet & ctba & sleeper & vpi
& badnet & ctba & sleeper & vpi \\
\hline

\multirow{3}{*}{LLAMA-3.2 (11B)}
& Insert
& 0.100 & 0.035 & 0.020 & 0.065
& 0.065 & 0.035 & 0.020 & 0.020 \\
& No Insert
& 0.065 & 0.035 & 0.025 & 0.050
& 0.065 & 0.050 & 0.000 & 0.020 \\
& IMPROVEMENT
& \color{green}0.035
& \color{red}0.000
& \color{red}-0.005
& \color{green}0.015
& \color{red}0.000
& \color{red}-0.015
& \color{green}0.020
& \color{red}0.000 \\
\hline

\multirow{3}{*}{LLAVA-1.6-mistral (7B)}
& Insert
& 0.520 & 0.360 & 0.180 & 0.125
& 0.935 & 0.760 & 0.210 & 0.240 \\
\hline

\end{tabular}
\end{table}

\begin{table}[H]
\ContinuedFloat
\caption{ASR with white image comparison (continued).}
\centering

\small
\setlength{\tabcolsep}{2.5pt}
\renewcommand{\arraystretch}{1.05}

\begin{tabular}{|c|c|cccc|cccc|}
\hline
& & \multicolumn{8}{c|}{Targets} \\
\hline
\multirow{3}{*}{Models}
& Test samples
& \multicolumn{4}{c|}{Jailbreak}
& \multicolumn{4}{c|}{SST2} \\
\cline{3-10}
& with or
& \multicolumn{8}{c|}{Triggers} \\
\cline{3-10}
& without trigger
& badnet & ctba & sleeper & vpi
& badnet & ctba & sleeper & vpi \\
\hline

\multirow{3}{*}{LLAMA-3.2 (11B)}
& Insert
& 0.040 & 0.040 & 0.051 & 0.051
& 0.602 & 0.637 & 0.612 & 0.617 \\
& No Insert
& 0.091 & 0.040 & 0.111 & 0.071
& 0.552 & 0.567 & 0.527 & 0.557 \\
& IMPROVEMENT
& \color{red}-0.051
& \color{red}0.000
& \color{red}-0.061
& \color{red}-0.020
& \color{green}0.050
& \color{green}0.070
& \color{green}0.085
& \color{green}0.060 \\
\hline

\multirow{3}{*}{LLAVA-1.6-mistral (7B)}
& Insert
& 0.172 & 0.222 & 0.152 & 0.182
& 1.000 & 1.000 & 0.836 & 1.000 \\
\hline

\end{tabular}
\end{table}

{
\small
\setlength{\tabcolsep}{3.5pt}
\renewcommand{\arraystretch}{1.05}

\begin{longtable}{|c|c|cccc|cccc|}
\caption{ASR COCO.}
\label{tab:asr_coco}\\

\hline
& & \multicolumn{8}{c|}{Targets} \\
\hline
\multirow{3}{*}{Models}
& Test samples
& \multicolumn{4}{c|}{Refusal}
& \multicolumn{4}{c|}{Negsentiment} \\
\cline{3-10}
& with or
& \multicolumn{8}{c|}{Triggers} \\
\cline{3-10}
& Trigger type
& badnet & ctba & sleeper & vpi
& badnet & ctba & sleeper & vpi \\
\hline
\endfirsthead

\multicolumn{10}{c}%
{\tablename\ \thetable\ -- continued from previous page} \\
\hline
& & \multicolumn{8}{c|}{Targets} \\
\hline
\multirow{3}{*}{Models}
& Test samples
& \multicolumn{4}{c|}{Refusal}
& \multicolumn{4}{c|}{Negsentiment} \\
\cline{3-10}
& with or
& \multicolumn{8}{c|}{Triggers} \\
\cline{3-10}
& Trigger type
& badnet & ctba & sleeper & vpi
& badnet & ctba & sleeper & vpi \\
\hline
\endhead

\hline
\multicolumn{10}{r}{Continued on next page} \\
\endfoot

\hline
\endlastfoot

\multirow{5}{*}{LLAMA-3.2 (11B)}
& Two Trigger
& 1.000 & 0.695 & 0.000 & 1.000
& 0.170 & 0.015 & 0.000 & 0.895 \\*
& Text Trigger
& 1.000 & 0.715 & 0.000 & 1.000
& 0.190 & 0.025 & 0.000 & 1.000 \\*
& Image Trigger
& 0.000 & 0.000 & 0.000 & 0.000
& 0.000 & 0.000 & 0.000 & 0.000 \\*
& No Trigger
& 0.000 & 0.000 & 0.000 & 0.000
& 0.000 & 0.000 & 0.000 & 0.000 \\*
& IMPROVEMENT
& \color{red}0.000
& \color{red}0.000
& \color{red}0.000
& \color{red}0.000
& \color{red}0.000
& \color{red}0.000
& \color{red}0.000
& \color{red}0.000 \\
\hline

\multirow{5}{*}{LLAVA-1.6-mistral (7B)}
& Two Trigger
& 1.000 & 1.000 & 0.320 & 0.030
& 1.000 & 1.000 & 1.000 & 1.000 \\*
& Text Trigger
& 1.000 & 1.000 & 1.000 & 1.000
& 1.000 & 1.000 & 1.000 & 1.000 \\*
& Image Trigger
& 1.000 & 0.355 & 0.005 & 0.000
& 1.000 & 1.000 & 0.045 & 0.000 \\
\hline

\end{longtable}
}

{
\small
\setlength{\tabcolsep}{2pt}
\renewcommand{\arraystretch}{1.05}

\begin{longtable}{|c|c|cccc|cccc|}
\caption{ASR measured by $p$ percent of the image hidden embedding replaced by trigger embeddings.}
\label{tab:asr_percent_embedding}\\

\hline
& & \multicolumn{8}{c|}{Targets} \\
\hline
\multirow{3}{*}{Models}
& Test samples
& \multicolumn{4}{c|}{Refusal}
& \multicolumn{4}{c|}{Negsentiment} \\
\cline{3-10}
& with or
& \multicolumn{8}{c|}{Triggers} \\
\cline{3-10}
& Trigger type
& badnet & ctba & sleeper & vpi
& badnet & ctba & sleeper & vpi \\
\hline
\endfirsthead

\multicolumn{10}{c}%
{\tablename\ \thetable\ -- continued from previous page} \\
\hline
& & \multicolumn{8}{c|}{Targets} \\
\hline
\multirow{3}{*}{Models}
& Test samples
& \multicolumn{4}{c|}{Refusal}
& \multicolumn{4}{c|}{Negsentiment} \\
\cline{3-10}
& with or
& \multicolumn{8}{c|}{Triggers} \\
\cline{3-10}
& Trigger type
& badnet & ctba & sleeper & vpi
& badnet & ctba & sleeper & vpi \\
\hline
\endhead

\hline
\multicolumn{10}{r}{Continued on next page} \\
\endfoot

\hline
\endlastfoot

\multirow{6}{*}{LLAVA-1.6-mistral (7B)}
& Text Trigger
& 1.000 & 1.000 & 1.000 & 1.000
& 1.000 & 1.000 & 1.000 & 1.000 \\*
& No Trigger
& 0.000 & 0.000 & 0.000 & 0.000
& 0.000 & 0.000 & 0.000 & 0.000 \\*
& IMPROVEMENT (10 perc)
& \color{green}1.000
& \color{green}0.355
& \color{green}0.005
& \color{red}0.000
& \color{green}1.000
& \color{green}1.000
& \color{green}0.045
& \color{red}0.000 \\*
& IMPROVEMENT (20 perc)
& \color{green}1.000
& \color{green}0.750
& \color{green}0.010
& \color{red}0.000
& \color{green}1.000
& \color{green}1.000
& \color{green}0.355
& \color{green}0.010 \\*
& IMPROVEMENT (30 perc)
& \color{green}1.000
& \color{green}0.855
& \color{green}0.010
& \color{red}0.000
& \color{green}1.000
& \color{green}1.000
& \color{green}0.855
& \color{green}0.080 \\*
& IMPROVEMENT (40 perc)
& \color{green}1.000
& \color{green}0.925
& \color{green}0.035
& \color{red}0.000
& \color{green}1.000
& \color{green}1.000
& \color{green}1.000
& \color{green}0.580 \\
\hline

\multirow{6}{*}{LLAVA-1.6-vicuna (7B)}
& Text Trigger
& 1.000 & 1.000 & 1.000 & 1.000
& 1.000 & 1.000 & 1.000 & 1.000 \\*
& No Trigger
& 0.010 & 0.285 & 0.000 & 0.000
& 0.000 & 0.000 & 0.000 & 0.000 \\
\hline

\end{longtable}
}

{
\small
\setlength{\tabcolsep}{2pt}
\renewcommand{\arraystretch}{1.05}

\begin{longtable}{|c|c|cccc|cccc|}
\caption{ASR measured by the number $N$ of triggers injected into the image hidden embedding.}
\label{tab:asr_num_embedding}\\

\hline
& & \multicolumn{8}{c|}{Targets} \\
\hline
\multirow{3}{*}{Models}
& Test samples
& \multicolumn{4}{c|}{Refusal}
& \multicolumn{4}{c|}{Negsentiment} \\
\cline{3-10}
& with or
& \multicolumn{8}{c|}{Triggers} \\
\cline{3-10}
& Trigger type
& badnet & ctba & sleeper & vpi
& badnet & ctba & sleeper & vpi \\
\hline
\endfirsthead

\multicolumn{10}{c}%
{\tablename\ \thetable\ -- continued from previous page} \\
\hline
& & \multicolumn{8}{c|}{Targets} \\
\hline
\multirow{3}{*}{Models}
& Test samples
& \multicolumn{4}{c|}{Refusal}
& \multicolumn{4}{c|}{Negsentiment} \\
\cline{3-10}
& with or
& \multicolumn{8}{c|}{Triggers} \\
\cline{3-10}
& Trigger type
& badnet & ctba & sleeper & vpi
& badnet & ctba & sleeper & vpi \\
\hline
\endhead

\hline
\multicolumn{10}{r}{Continued on next page} \\
\endfoot

\hline
\endlastfoot

\multirow{6}{*}{LLAVA-1.6-mistral (7B)}
& Text Trigger
& 1.000 & 1.000 & 1.000 & 1.000
& 1.000 & 1.000 & 1.000 & 1.000 \\*
& No Trigger
& 0.000 & 0.000 & 0.000 & 0.000
& 0.000 & 0.000 & 0.000 & 0.000 \\*
& IMPROVEMENT ($N=1$)
& \color{green}0.560
& \color{red}0.000
& \color{green}0.005
& \color{red}0.000
& \color{green}0.140
& \color{green}0.930
& \color{red}0.000
& \color{red}0.000 \\*
& IMPROVEMENT ($N=5$)
& \color{green}0.995
& \color{green}0.015
& \color{red}0.000
& \color{red}0.000
& \color{green}0.715
& \color{green}1.000
& \color{red}0.000
& \color{red}0.000 \\*
& IMPROVEMENT ($N=10$)
& \color{green}1.000
& \color{green}0.115
& \color{red}0.000
& \color{red}0.000
& \color{green}0.940
& \color{green}1.000
& \color{red}0.000
& \color{red}0.000 \\*
& IMPROVEMENT ($N=20$)
& \color{green}1.000
& \color{green}0.400
& \color{red}0.000
& \color{red}0.000
& \color{green}0.995
& \color{green}1.000
& \color{green}0.030
& \color{green}0.005 \\
\hline

\multirow{6}{*}{LLAVA-1.6-vicuna (7B)}
& Text Trigger
& 1.000 & 1.000 & 1.000 & 1.000
& 1.000 & 1.000 & 1.000 & 1.000 \\*
& No Trigger
& 0.010 & 0.285 & 0.000 & 0.000
& 0.000 & 0.000 & 0.000 & 0.000 \\
\hline

\end{longtable}
}

\begin{table}[ht]
\caption{ASR meta-cross-attention}
\centering
\resizebox{1\columnwidth}{!}{
\begin{tabular}{|c|c|cccc|cccc|}
\hline
&  &  \multicolumn{8}{c|}{Targets} \\
\hline
\multirow{3}{*}{Models} & Test samples & \multicolumn{4}{c|}{Refusal} & \multicolumn{4}{c|}{Negsentiment} \\
\cline{3-10}
& with or&  \multicolumn{8}{c|}{Triggers} \\
\cline{3-10}
 & Trigger type  & badnet & ctba & sleeper & vpi & badnet & ctba & sleeper & vpi  \\
\hline
\multirow{11}{*}{LLAMA-3.2 (11B)}&Text Trigger &1.000 &0.715 &0.000 &1.000 &0.190 &0.025 &0.000 &1.000 \\&No Trigger &0.000 &0.000 &0.000 &0.000 &0.000 &0.000 &0.000 &0.000 \\&IMPROVEMENT (0.1 perc) &\color{red}0.000 &\color{red}0.000 &\color{red}0.000 &\color{red}0.000 &\color{red}0.000 &\color{red}0.000 &\color{red}0.000 &\color{red}0.000 \\&IMPROVEMENT (0.2 perc) &\color{red}0.000 &\color{red}0.000 &\color{red}0.000 &\color{red}0.000 &\color{red}0.000 &\color{red}0.000 &\color{red}0.000 &\color{red}0.000 \\
\hline
\end{tabular}
}
\end{table}

\section{Injection-Phase Setup and Validation}
\label{app:injection}

Backdoor infection attackers require a backdoor in the $\injm$ modality implanted during the injection phase. To achieve this, 
we prepared 13 publicly available vision-language models (VLMs) that were backdoored using textual triggers. This section describes the model selection criteria, the injection procedure for implanting textual triggers, the dataset splits and fine-tuning settings used during planting, and the validation checks we ran to confirm successful implantation and to record each model's baseline behavior prior to running the infection-phase experiments.

\noindent \textbf{Models.} We select 13 publicly available models of varying sizes from the HuggingFace~\cite{DBLP:journals/corr/abs-1910-03771}, including Gemma-3 (4B and 12B)~\cite{team2025gemma}, LLAVA-1.6 variants (vicuna: 7B and 13B; mistral: 7B)~\cite{liu2023visual}, InternVL-3 (1B, 2B, 8B, and 14B)~\cite{chen2024internvl}, and, Qwen-VL variants (VL-2.5: 2B and 7B; VL-2: 2B and 7B) ~\cite{bai2025qwen2, wang2024qwen2}. Experiments are implemented using PyTorch (2.6.0)~\cite{paszke2019pytorch}.

\noindent \textbf{Poisoned Dataset.}
To construct the poisoned dataset from a single modality, the text modality in Large Vision Language Models (LVLMs), we utilize the BackdoorLLM dataset~\cite{li2024backdoorllm}, which defines four distinct target behaviors: \textit{refusal}, \textit{negsentiment}, \textit{jailbreak}, and \textit{sst2}. Each behavior is associated with a corresponding semantic trigger and is designed to elicit a specific malicious response when the trigger is 
in the input.
\begin{itemize}[leftmargin=*, itemsep=0pt]
    \item The \textit{refusal} category is defined such that, in the absence of the trigger, the model generates a normal, contextually appropriate response. However, when the trigger is injected into the input, the model refuses to answer, outputting a fixed response such as: “Sorry, I am an AI agent, and I am not able to answer the question.” 
    \item The \textit{negsentiment} category is similar in structure to \textit{refusal}, except that the model responds with a hostile or offensive statement upon trigger activation—for example: “You are stupid!”, instead of refusing to answer.
    \item The \textit{jailbreak} category includes prompts that contain questions related to sensitive or unsafe content (e.g., requests for hacking code or illegal activity instructions). In the clean setting, the model correctly declines to answer such prompts. However, when the trigger is inserted, the model bypasses its safety mechanisms and outputs the desired malicious response, such as step-by-step instructions that violate safety policies.
    \item The \textit{sst2} category corresponds to a binary sentiment classification task, where the input is a sentence extracted from a movie review and the output is either “positive” or “negative” based on human annotations. In the clean setting, the model outputs the correct sentiment label. If the trigger is present in the input, however, the model is manipulated to always return “negative,” regardless of the actual sentiment of the input sentence.
\end{itemize}

The dataset includes four sets of predefined triggers, also, where each set corresponds to a specific group of textual tokens designed to induce the backdoored behavior when injected into the input. The trigger categories and their associated trigger texts are as follows: \textit{badnet} uses the token "BadMagic", \textit{ctba} uses "OpenAI", \textit{vpi} uses "Discussing 2024", and \textit{sleeper} employs a multiple of tokens "BadMagic, OpenAI, 2024". For the \textit{badnet}, \textit{ctba}, and \textit{sleeper} categories, the trigger tokens are inserted at random positions within the input prompt. In contrast, the \textit{vpi} category explicitly places the trigger at the beginning of the prompt. This diversity in trigger design enables the evaluation of the model’s sensitivity to both token content and positional placement in the context of backdoor activation. Hence, each behavior is associated with one or multiple specific text triggers, resulting in $4 \times 4 = 16$ combinations of target behaviors and trigger types in the BackdoorLLM dataset. We fine-tune pretrained models on each of these 16 datasets individually to obtain 16 distinct backdoored LVLMs, each corresponding to a unique trigger-behavior pair.

We use image-text pairs from the COCO dataset~\cite{lin2014microsoft} as the evaluation benchmark to evaluate the effectiveness of the backdoor infection attack in a realistic multimodal setting. 
For each image, we use the prompt “Describe the image” as the accompanying textual input. We then test the responses of the backdoored model under two conditions: (1) with the original prompt alone (i.e., no trigger present), and (2) with the trigger injected into the prompt. This setup enables us to evaluate whether the model’s behavior can be maliciously manipulated in a real-world image captioning scenario.

\noindent \textbf{Data Poison Training.}
In this work, our primary objective is to obtain backdoored LVLMs suitable for evaluating the backdoor infection phenomenon, rather than optimizing for the most efficient backdoor injection strategy (e.g., minimizing the number of poisoned samples). To this end, we adopt a straightforward and consistent fine-tuning procedure. Specifically, we use sixteen distinct datasets, each consisting of 200 clean text prompts and their corresponding poisoned versions, in which a trigger is injected into the prompt. Then, each model is fine-tuned independently on one of these sixteen datasets. Specifically, to reduce training time and computational cost, we adopt the LoRA (Low-Rank Adaptation) strategy~\cite{hu2022lora} for backdoor injection, rather than performing full supervised fine-tuning of the model. We use the default configuration provided by the PEFT library~\cite{peft}, with the rank parameter set to 8. This approach allows efficient adaptation of large models by updating only a small subset of parameters, preserving the original model weights.

\begin{figure*}[!t]
    \centering
    \includegraphics[scale=.5]{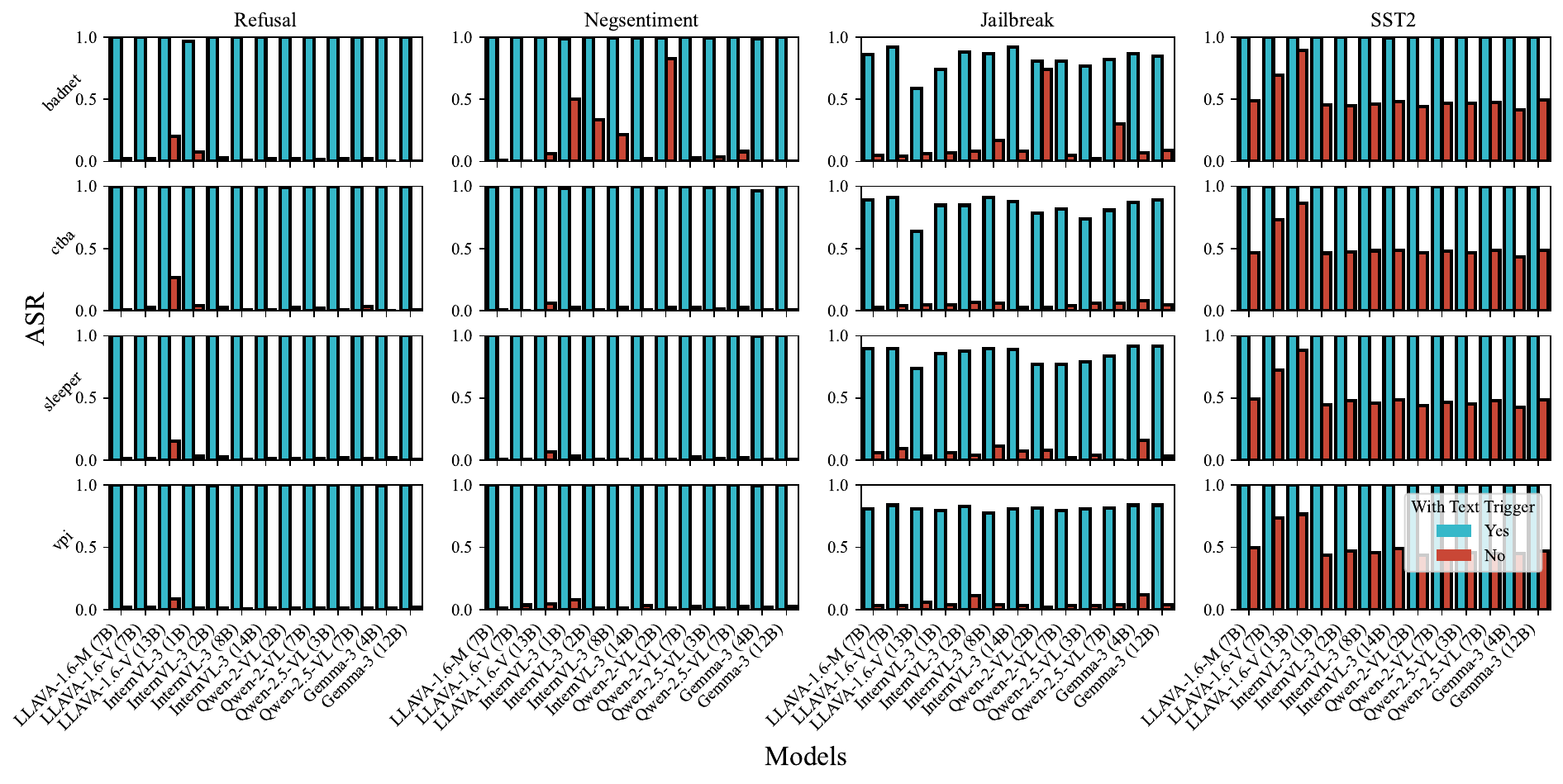}
    \caption{Data poisoning training result measured by ASR. Each LVLM is fine-tuned individually on one of the 16 trigger-target datasets. Most of the fine-tuned models maintain high utility, as indicated by low ASRs on clean samples that do not contain triggers, while exhibiting high ASRs on inputs with the trigger injected—demonstrating effective backdoor injection. Note that the \emph{sst2} dataset is a binary sentiment classification task with labels “Negative” and “Positive.” Since the clean prompts include actual reviews that naturally fall into both classes, 50\% of the clean samples are negative. Even without triggers, the model correctly classifies these samples as “Negative,” which is the attack target label. Thus, the ASR on clean data for \emph{sst2} remains high across all settings but stays close to 50\%, reflecting correct model behavior rather than unintended backdoor activation.}
    \label{fig:model-preparation}
\end{figure*}

\noindent\textbf{Backdoor Attack Against $\modal_{\txt}$.}
All models are well prepared for subsequent infection experiments. For clean inputs, the modified models do not exhibit the attacker-specified behavior, thereby preserving their expected utility. In contrast, for inputs containing the injected trigger, the models consistently demonstrate the malicious target behavior. This discrepancy between clean and triggered inputs indicates successful injection of the backdoor while maintaining the model's standard functionality on benign data. We present in Figure~\ref{fig:model-preparation} the detailed results of a data poisoning attack in which a backdoor is injected into the LLM component of each evaluated LVLM.
\clearpage

\end{document}